\pdfoutput=1

\documentclass[sigconf, nonacm]{acmart}

\newcommand\vldbdoi{10.14778/3819518.3819551}
\newcommand\vldbpages{2289 - 2302}
\newcommand\vldbvolume{19}
\newcommand\vldbissue{9}
\newcommand\vldbyear{2026}
\newcommand\vldbauthors{\authors}
\newcommand\vldbtitle{\shorttitle} 
\newcommand\vldbavailabilityurl{https://github.com/zhangzhongshuai/C-KMaxI-code}
\newcommand\vldbpagestyle{empty}

\usepackage{graphicx}
\usepackage{balance}
\usepackage[linesnumbered,ruled,vlined]{algorithm2e}
\usepackage{subfigure}

\usepackage{amssymb}
\usepackage{comment}
\usepackage{textcomp}
\usepackage{graphics}
\usepackage{amsmath}
\usepackage{tabularx} 
\usepackage{array}

\usepackage{enumerate}
\usepackage{amsmath}
\usepackage{color,xcolor}
\usepackage{soul}
\usepackage[linesnumbered,ruled,vlined]{algorithm2e}
\usepackage{float}
\usepackage{engord}
\usepackage{multirow}
\usepackage{amsthm}
\usepackage{algpseudocode}
\usepackage{enumitem}
\setlist[itemize]{noitemsep, topsep=0pt}
\renewcommand{\proofname}{\textbf{Proof.}}
\let\oldproof\proof
\renewcommand{\proof}{\oldproof\setlength{\parindent}{2em}} 
\newtheorem{definition}{Definition}

\newtheorem{lemma}{Lemma}

\newtheorem{theorem}{Theorem}

\newcommand{\query}{C-KMaxI}
\everymath{\displaystyle}

\begin{document}
\title{Continuous Query for Top-$K$ Maximal Sum Intervals over Streaming Data}

\author{Zhongshuai Zhang}
\affiliation{%
  \institution{Beijing Institute of Technology}
}
\email{zhang\_zhongshuai@126.com}

\author{Xiaochun Yang}
\affiliation{%
  \institution{Northeastern University}
}
\email{yangxc@mail.neu.edu.cn}

\author{Baihua Zheng}
\affiliation{%
  \institution{Singapore Management University}
}
\email{bhzheng@smu.edu.sg}

\author{Rui Zhu}
\affiliation{%
  \institution{Shenyang Aerospace University}
}
\email{zhurui@sau.edu.cn}

\author{Haomin Li}
\affiliation{%
  \institution{Northeastern University}
}
\email{lihm@mails.neu.edu.cn}

\author{Bin Wang}
\affiliation{%
  \institution{Northeastern University}
}
\email{binwang@mail.neu.edu.cn}

\begin{abstract}
The continuous identification of top-$k$ maximal sum intervals using a sliding window over a data stream is a critical operation for applications in IoT and beyond. A maximal sum interval is a non-overlapping, contiguous subsequence with the maximal sum in a sequence of signed values. Existing algorithms are ill-suited for streaming contexts: they either exhaustively enumerate all intervals even for small $k$ values, or depend on indexes that require frequent and costly restructuring. 
We propose a novel \emph{partition}-based strategy. Our core insight is a partitioning scheme that guarantees that any maximal sum interval is fully contained within a single partition, enabling independent and parallel processing. This design provides two key advantages: it enables safe pruning of partitions that cannot contribute to top-$k$ results, drastically narrowing the search space, and it enables efficient, incremental maintenance of the maximal sum intervals in each partition. 
We develop algorithms for partition construction, incremental partition updates, and partition-based top-$k$ maximal sum interval search. 
Extensive experiments on real and synthetic datasets demonstrate that our approach significantly improves efficiency.
\end{abstract}

\maketitle

\begingroup
\renewcommand\thefootnote{}\footnote{\noindent
This paper has been accepted by VLDB 2026.}
\addtocounter{footnote}{-1}
\endgroup

\pagestyle{\vldbpagestyle}
\begingroup\small\noindent\raggedright\textbf{PVLDB Reference Format:}\\
\vldbauthors. \vldbtitle. PVLDB, \vldbvolume(\vldbissue): \vldbpages, \vldbyear.\\
\href{https://doi.org/\vldbdoi}{doi:\vldbdoi}
\endgroup
\begingroup
\renewcommand\thefootnote{}\footnote{\noindent
This work is licensed under the Creative Commons BY-NC-ND 4.0 International License. Visit \url{https://creativecommons.org/licenses/by-nc-nd/4.0/} to view a copy of this license. For any use beyond those covered by this license, obtain permission by emailing \href{mailto:info@vldb.org}{info@vldb.org}. Copyright is held by the owner/author(s). Publication rights licensed to the VLDB Endowment. \\
\raggedright Proceedings of the VLDB Endowment, Vol. \vldbvolume, No. \vldbissue\ %
ISSN 2150-8097. \\
\href{https://doi.org/\vldbdoi}{doi:\vldbdoi} \\
}\addtocounter{footnote}{-1}\endgroup

\ifdefempty{\vldbavailabilityurl}{}{
\begingroup\small\noindent\raggedright\textbf{PVLDB Artifact Availability:}\\
The source code, data, and/or other artifacts have been made available at \url{\vldbavailabilityurl}. 
\endgroup
}

\section{Introduction}\label{sec:intro}

The Top-$k$ Maximal Sum Intervals  (KMaxI) problem~\cite{DBLP:conf/ismb/RuzzoT99, DBLP:journals/ijfcs/BaeT07}, a classic problem studied for over 25 years, aims to find the $k$ non-overlapping, contiguous subsequences with the highest sums in a sequence of real numbers. These intervals are ranked such that the $i$-th interval has the maximum possible sum from the elements disjoint from all higher-ranked (1st to ($i-1$)-th) intervals. Identifying such high-impact contiguous segments is fundamental in time-series analysis, supporting applications such as burst detection, anomaly identification, and trend discovery in sequential data.

This paper extends KMaxI to the challenging streaming context. 
In many modern applications, signals arrive continuously and must be analyzed in real time, making it necessary to support efficient interval discovery over sliding windows.
To address this need, we introduce the Continuous KMaxI query (\query), which provides real-time insights over a sliding window. Sliding windows can be time- or count-based~\cite{DBLP:conf/sigmod/MouratidisBP06, DBLP:conf/cikm/HaghaniMA09, DBLP:conf/bigdataconf/RamperezZV21, DBLP:journals/vldb/JinYCYL10, DBLP:conf/stoc/AjtaiJKS02, DBLP:conf/icde/GedikWYL07, DBLP:conf/icde/GuYW07, DBLP:journals/pvldb/SubercazeGGKL17, DBLP:conf/edbt/Li0ZMY17, DBLP:journals/dase/ChenLYL17}. In this work, we focus on the count-based model. 
Formally, a \query\ query, denoted as $q(n,s,k)$, continuously monitors the $n$ most recent signed numeric values in the query window $W$. Each time $W$ advances by $s$ positions, the query returns the top-$k$ maximal sum intervals within the updated window. The techniques developed in this work can also be applied to time-based sliding windows. 

\begin{figure}[t]
    \centering
    \includegraphics[width=\linewidth]{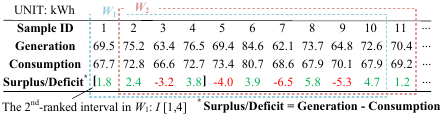}   
            \caption{A running example on power generation, consumption, and storage of a microgrid.}
    \label{fig:running example}
\end{figure}

\noindent
\textbf{Motivations and Challenges}. The \query\ extends interval-based analysis to streaming environments by continuously identifying high-impact contiguous periods within the sliding window. Such intervals capture sustained periods of strong cumulative signal values and are essential for detecting bursts, trends, or anomalous behaviors. In many real-world applications, identifying multiple disjoint high-impact periods is more informative than detecting a single peak interval, as systems often need to prioritize several candidate periods for downstream analysis or operational planning.

The \query\ query is motivated by its importance in real-time analytics across diverse domains. For example, it can be used  to identify significant genomic regions in high-throughput bio-informatics sequencing~\cite{DBLP:conf/ismb/RuzzoT99, doi:10.1073/pnas.90.12.5873, doi:10.1126/science.1621093, spouge2012ruzzo}, detect trending patterns in web scraping~\cite{10.1145/1526709.1526840} and information retrieval systems~\cite{10.1145/2661829.2661905}, and extract informative temporal features for online AI models~\cite{jin2025ideally}. 
In these applications, the data naturally form signed sequences, where values may be positive or negative, representing gains and losses, increases and decreases, or surplus and deficit signals.
While KMaxI handles real-valued sequences, many practical signals are signed. Supporting signed sequences therefore enables systems to identify sustained positive trends despite temporary fluctuations.

A compelling application of {\query} query arises in the Internet of Things (IoT), particularly in real-time monitoring and energy management for microgrids, which have received significant attention in recent years~\cite{Escobar2024,10999049}. Figure~\ref{fig:running example} illustrates a running example.  
Here, a \query\ query continuously analyzes a \emph{signed} microgrid signal to identify time periods with the largest \emph{cumulative} energy surplus.
Specifically, the input stream can be modeled as a surplus/deficit sequence, where each value represents the difference between power generation and consumption. Positive values indicate energy surplus, while negative values indicate energy deficit. 
The returned top-$k$ non-overlapping maximal sum intervals therefore correspond to the $k$ most significant, disjoint periods with the highest cumulative surplus. These intervals correspond to sustained periods of excess generation, enabling operators to identify when surplus energy should be stored, redistributed, or sold to the grid. While we illustrate the query using a microgrid monitoring scenario, the proposed \query\ framework is broadly applicable to many streaming analytics tasks that require identifying sustained high-impact periods  from continuously evolving signals.

As new measurements arrive, the query result evolves with the sliding window. Here, $I[i,j]$ denotes the contiguous interval from
objects $o_i$ to $o_j$ and its associated sum is the total value over that interval.
For example, the top intervals may change from $I[8,8]$ with sum $5.8$ and $I[1,4]$ with sum $4.8$ in window $W_1$, to $I[8,11]$ with sum $6.4$ and $I[6,6]$ with sum $3.9$ in window $W_2$. This dynamic update capability enables continuous monitoring of emerging surplus periods as the data stream evolves.

Despite its practical importance, efficiently computing \query\ queries over data streams presents two formidable challenges. 
First, the combination of high-velocity data streams and the inherent computational complexity of the KMaxI problem makes it difficult to deliver low-latency results, especially for large window sizes. Second, the set of top-$k$ intervals is highly volatile. As the window slides, intervals may split, merge, or be completely replaced. For instance, as shown in Figure~\ref{fig:running example}, when the window advances to $W_2$, the result interval $I[8,8]$ merges with adjacent intervals to form a new result interval $I[8,11]$. This volatility renders static algorithms inefficient, as they cannot update results incrementally and must instead perform costly re-computations from scratch.

\noindent
\textbf{State-Of-The-Art Efforts}. 
Existing research on the KMaxI problem is primarily designed for static data environments and falls into two categories.  Scan-based approaches, such as Ruzzo-Tompa~\cite{DBLP:conf/ismb/RuzzoT99}, identify all maximal intervals through a sequential scan but require a subsequent sorting step to select the top-$k$ results, which becomes inefficient especially when $k$ is much smaller than the total count. 
Index-based approaches, like the Tournament-Tree~\cite{DBLP:journals/ijfcs/BaeT07}, construct auxiliary structures that allow direct retrieval of the top-$k$ intervals without exhaustive comparisons. While effective in static scenarios, these approaches are not suitable for data streams. For example, adapting the Tournament-Tree would necessitate frequent and costly index reconstructions with every window slide, resulting in prohibitive computational overhead. Consequently, existing algorithms cannot \emph{efficiently} support the \query\ query.

\noindent
\textbf{Our Solution}. To overcome these limitations, we propose a novel partition-based framework for efficient {\query} processing. 
The key idea is to partition the query window into a set of 
disjoint partitions 
such that every maximal sum interval is fully contained within a single partition. This property enables two important advantages. 
First, since the top-$k$ results must come from at most $k$ partitions, many partitions can be safely pruned, 
drastically reducing the search space. Second, changes in maximal sum intervals are closely tied to the evolution of partitions as the window slides, which can be categorized into four distinct patterns, namely unaltered, shrinking, extended, and emergent. We develop efficient update strategies for each pattern, maximizing the reuse of previously identified maximal sum intervals.

Based on this insight, our solution proceeds in three stages. First, we construct the partitions and identify the top-1 maximal sum interval within each partition using a dynamic programming strategy. 
Next, we perform a partition-based $k$-maximal sum interval search
algorithm that examines promising partitions (at most $k$ partitions) to identify the top-$k$ results. Finally, we design an incremental maintenance mechanism that updates partitions and their corresponding maximal intervals as the query window evolves. By carefully categorizing partition evolution patterns and reusing previously computed results whenever possible, our approach significantly reduces the computational overhead of \query\ evaluation.

\renewcommand{\arraystretch}{0.75}
\begin{table}[t]
\footnotesize
  \caption{Notations}
  \label{tab:Notations}
  \centering
  \begin{tabularx}{\columnwidth}{|>{\centering\arraybackslash}m{0.17\columnwidth}|X|}
    \hline
    {\centering Symbols} & {\centering Description \par} \\ \hline
    $o_i$ & the $i$-th real number object \\ \hline
    $I[i,j]$, $S[i,j]$ & the interval from $o_i$ to $o_j$, the sum of objects in $I[i,j]$ \\ \hline
    $MaxI[i,j]$ & maximal sum interval from $o_i$ to $o_j$ \\ \hline
    $q(n,s,k)$ & the query, modeled as a sliding window $W$ monitoring most recent $n$ objects, returning the top-$k$ maximal intervals over a data stream $\mathcal{S}$, with $s$ being object update number for each slide \\ \hline
    $M[i]$, $M[\cdot]$ & the largest sum of an interval starting at position $i$, $M[\cdot]$ is the collective term for $M_i$ \\ \hline
    $l_j$, $r_j$, $LM_j$ & the leftmost index in partition $P_j$, the rightmost index in $P_j$, the index of the largest $M$ in $P_j$ \\ \hline
    pTMS $M[u,v]$ & a Potential Target $M$-Subsequence that has $M[u]/M[v]$ as its strict maximum/minimum $M$ value. \\ \hline
  \end{tabularx}
\end{table}

\section{Problem Definition} \label{sec:problem def}

A \emph{maximal sum interval} (or \emph{maximal interval}, denoted as $MaxI[i,j]$), is a contiguous subsequence $(o_i,\cdots,o_j)$ in a sequence of real numbers $(o_1,o_2,\cdots,o_n)$ whose sum $S[i,j]$ is global maximal, i.e., $S[i,j]=\max\nolimits_{1\le i\le j\le n} S[i,j]$, with $S[i,j]>0$.
This paper considers the set of globally ranked, positive, and non-overlapping maximal intervals, following the iterative definition in~\cite{DBLP:conf/ismb/RuzzoT99}. Specifically, the interval with the largest sum in the entire sequence
is ranked first; the $i$-th interval is the maximal interval computed from elements disjoint from all higher-ranked intervals.
To avoid ambiguity under ties, we treat intervals containing non-empty zero-sum prefixes or suffixes as invalid.
Otherwise, such intervals could create multiple overlapping intervals with identical sums when adjacent to positive-sum segments. Therefore, if a sequence has fewer than $k$ valid maximal intervals, \query\ may return fewer than $k$ results.

Consider the sequence in query window $W_1$ in Figure~\ref{fig:problem definition}. 
The interval $I[10,10]$ (sum=12) is a global maximum sum interval since no other subsequence has a larger sum. 
Subsequences like $I[3,3]$ and $I[11,13]$ are invalid due to non-positive sums. While $I[1,2]$ and $I[2,2]$ both have sum 6, $I[1,2]$ is invalid as it contains a zero-sum prefix. In addition, each lower rank interval must be disjoint from all higher-ranked ones, e.g., $I[4,10]$ has a sum larger than $I[4,8]$ but is inadmissible, as it overlaps with the 1st-ranked interval $I[10,10]$.

We now formally define \textbf{\query}, the \underline{c}ontinuous query for top-\underline{$k$} \underline{max}imal sum \underline{i}ntervals over streaming data. 

\begin{definition}  \label{defn:query definition}
\textbf{\query.} 
Given a stream $\mathcal{S}$ of 
signed values (hereafter referred to as objects), 
a \query\ query, denoted as $q(n,s,k)$, maintains a sliding window $W$ of the most recent $n$ objects from $\mathcal{S}$. Each time the window slides by $s$ positions ($s$ objects expire and $s$ new objects arrive), the query returns the top-$k$ maximal intervals $\{MaxI_1, MaxI_2, \dots, MaxI_k\}$ in the updated window $W$. 
\end{definition}

Figure~\ref{fig:problem definition} illustrates an example \query\ query $q(13,4,2)$. In window $W_1$, four maximal intervals exist: $MaxI[2,2]$, $MaxI[4,8]$, $MaxI[10,10]$, and $MaxI[12,12]$. The top-2 results are $MaxI_1[10,10]$ (sum = $12$) and $MaxI_2[4,8]$ (sum=$10$). 
After the window slides to $W_2$, the intervals update to $MaxI'[6,6]$, $MaxI'[8,8]$, $MaxI'[10,10]$, and $MaxI'[12,17]$. The top-2 results become $MaxI_1'[12,17]$ (sum=$13$) and $MaxI_2'[10,10]$ (sum=$12$).

Streams with signed values naturally arise in many real-world applications. Examples include surplus signals in energy systems (supply$-$demand), financial price changes, sensor deltas, and residual signals obtained after subtracting a baseline. Residual signals, in particular, convert strictly positive or negative measurements into signed streams (e.g., by subtracting  the mean or expected value). In all these scenarios, a \query\ query identifies the top-$k$ intervals with the largest cumulative deviations.

\begin{figure}[t]
    \centering
            \includegraphics[width=\linewidth]{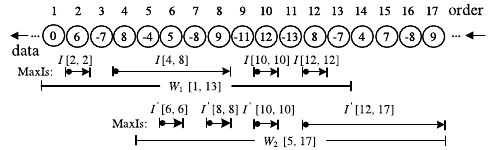}   
            \caption{An example of $q(n=13, s=4, k=2)$.}
    \label{fig:problem definition}
\end{figure}

\section{Partitioning}\label{sec:partition-defn}

This section formalizes our partitioning strategy for optimizing \query\ queries. The core idea of this approach is a guarantee that any maximal interval lies entirely within a single partition, which enables independent processing of partitions. This property yields two key advantages: safe pruning of partitions that cannot contain result intervals, and efficient incremental maintenance via partition evolution patterns. We begin by formally defining partitions and subsequently demonstrate how these characteristics are realized.

\begin{definition} \label{defn:partition}
\textbf{Partition.} A partition $\mathcal{P}=\{P_1, P_2,\dots, P_p\}$ of a query window $W=(o_1,o_2,\dots,o_n)$  is a division of  $W$ into contiguous, non-overlapping segments that collectively cover $W$.  Each partition \(P_i = \{o_{l_i}, o_{l_i+1}, \dots, o_{r_i}\}=[l_i,r_i]\) with \(1\le l_i \leq r_i\leq n\) has a right boundary at position $r_i$, which is defined as the position that maximizes the cumulative sum from the start of the partition $P_i$; that is, $S[l_i, r_i]\ge S[l_i, t]$ for all $t\in [l_i, n]$.  
\end{definition}

Figure~\ref{fig:benefit of incremental maintain-v2} shows the partitioning of $W_1=(o_1, \cdots, o_{19})$ into five partitions, $P_1,\ldots,P_5$.
Consider partition $P_1=(o_1, o_2, o_3, o_4)=[1,4]$.  Its right boundary is at position 4 (i.e., $r_1=4$). This means that for a subsequence starting from position 1, the cumulative sum is maximized at position 4; any interval $[1, j]$ with $j\ne 4$ yields a smaller cumulative sum. This property extends to any subsequence starting from $s\in [1,4]$, as formally established in Lemma~\ref{lemma:right-bound alignment}. This follows directly from the functional definition of the partition's right boundary. We omit the formal proof in the interest of brevity.

\begin{figure}[t]
    \centering
            \includegraphics[width=\linewidth]{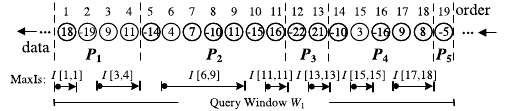}   
            \caption{An example of partitions. 
            }
    \label{fig:benefit of incremental maintain-v2}
\end{figure}

\begin{lemma}\label{lemma:right-bound alignment}
\textbf{Right-bound Alignment.} Let partition $P_i=(o_{l_i}, o_{l_i+1}$, $\dots, o_{r_i})=P[l_i,r_i]$ be a partition of size greater than 1 of window  $W=(o_1, o_2$, $\cdots, o_n)$, where $r_i$ is its right boundary as defined in Definition~\ref{defn:partition}. Then, for every start index $j\in [l_i, r_i]$, the cumulative sum starting from $j$ also attains its maximum at the same right boundary $r_i$. Formally, $\forall j\in [l_i,r_i]$, $S[j,r_i]=\max\nolimits_{t\in[j,n]}S[j,t]$.  
\end{lemma}

Beyond the right-bound alignment property, our partitioning scheme exhibits a more profound property: any maximal interval is entirely contained within a single partition. This guarantee is formally encapsulated by Theorem~\ref{theo:mii overlaps}.

\begin{theorem} \label{theo:mii overlaps}
No maximal interval overlaps with any two adjacent partitions.
\end{theorem}

\begin{proof} We proceed by contradiction. Assume the theorem is false, i.e., meaning there exists a maximal interval $MaxI[a,b]$ that overlaps with two adjacent partitions $P_i=[l_i,r_i]$ and $P_{i+1}=[l_{i+1},r_{i+1}]$ where $r_i=l_{i+1}-1$, $l_i\le a \le r_i$, and $b\ge l_{i+1}$. By the definition of a maximal interval, we have $S[a,b]> S[a,r_i]$. On the other hand, the right-bound alignment property of a partition guarantees that  $S[a,b]\leq S[a,r_i]$. These two constraints introduce a contradiction, implying that our initial assumption is false, thereby proving the theorem. 
\end{proof}

For $W_1$ in Figure~\ref{fig:benefit of incremental maintain-v2}, there are in total seven maximal intervals:
$\{MaxI[1,1]$, $MaxI[3,4]$, $MaxI[6,9]$, $MaxI[11,11]$, $MaxI[13,13]$, $MaxI[15,15]$, $MaxI[17,18]$\}. As guaranteed by Theorem~\ref{theo:mii overlaps}, each maximal interval is contained entirely within a single partition (e.g. $MaxI[1,1]$ and $MaxI[3,4]$ lie in partition $P_1$).

A key benefit of our partitioning strategy is its support for efficient query processing by focusing on local maximal intervals. Since Theorem~\ref{theo:mii overlaps} guarantees that any maximal interval is entirely contained within a single partition, partitions can be processed independently. This containment enables a crucial pruning heuristic: we first identify the top-1 maximal interval within each partition. If this local champion for a given partition $P_i$ does not qualify for the global top-$k$ list, then no other intervals within $P_i$ can be query results. Therefore, the partition $P_i$ can be safely pruned.

Consider the query in Figure~\ref{fig:benefit of incremental maintain-v2} that searches for top-$3$ maximal intervals within $W_1$. The top-1 maximal intervals with each partition are $MaxI[3,4]$ (sum $20$) in $P_1$, $MaxI[11,11]$ (sum $16$) in $P_2 $, $MaxI[13,13]$ (sum $21$) in $P_3$, and $MaxI[17,18]$ (sum $17$)  in $P_4$. $P_5$, only containing a non-positive object, can be skipped as it does not contain any maximal intervals. 
Among these candidates, $MaxI[11,11]$ from $P_2$ is not among the top-$3$. Consequently, no other interval in $P_2$, such as $MaxI[6,9]$, can be a query result, and the entire partition can be pruned.

\section{The Initialization Construction}
In this section, we first introduce the partition construction algorithm and then propose a search algorithm to locate $k$-maximal intervals based on partitions.

\subsection{The Partition Construction}\label{sec:partition-opt}

We propose a dynamic programming-based partition construction algorithm. 
Moreover, while constructing the partitions, it provides direct access to the top-1 maximal interval within each partition, without requiring extra computation. 

Given a query window $W=(o_1, \cdots, o_n)$, we define two arrays $M[1 \cdots n]$ and $E[1 \cdots n]$. Here, $M[i]$ stores the largest cumulative sum of any contiguous interval starting at position $i$, and $E[i]$ records the ending position of the corresponding interval, i.e., $S[i,E[i]]=M[i]$. Formally, 
$M[i]= S[i, E[i]] = \mathop{\max}\limits_{1\leq i \leq j \leq n} \sum\nolimits_{\mu=i}^j o_\mu$.
Both arrays can be computed in reverse order using the following recurrence relation: 
\begin{equation}
\label{equ:dp}
   ( M[i], E[i]) =
  \begin{cases}
  (o_i+M[i+1], E[i+1]), & \text{if } M[i+1]>0, \\[3pt]
  (o_i, i), & \text{if } M[i+1]\leq 0
  \end{cases}  
\end{equation}
with the base case $M[n]=o_n$ and $E[n]=n$.  
Intuitively, if $M[i+1]\le 0$, the single-object interval $I[i,i]$ (i.e, just $o_i$) yields the maximal sum starting at $i$, as including any subsequent interval would only decrease the sum. 
Crucially, any position $i$ where $M[i]\le 0$ is identified as a \emph{cutting point}, as formalized in Definition~\ref{defn: cutting Point}. A non-positive $M[i]$ indicates that no interval $I[j,\cdot]$ starting at position $j<i$ 
can achieve a larger sum by extending to include any interval starting at position $i$. Consequently, these positions naturally form the boundaries of our partitions. 

\begin{definition} \label{defn: cutting Point}
    \textbf{Cutting Points.} A cutting point is a position $i$ in the query window $W$ for which the value $M[i]\le 0$. The object $o_i$ located at a cutting point is the physical marker for that boundary.
\end{definition}

\begin{figure}[t]
    \centering
            \includegraphics[width=\linewidth]{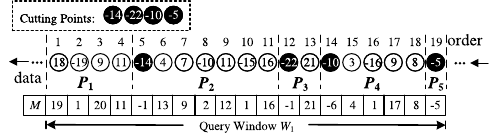}   
            \caption{Constructing partitions.}
    \label{fig:constructing partitions}
\end{figure}

A partition is determined by two consecutive cutting points, inclusive of its left cutting point and exclusive of the right one. The window boundaries (i.e., positions 1 and $n$ of a query window of size $n$) serve as implicit partition boundaries if they are not themselves cutting points.
As illustrated in Figure~\ref{fig:constructing partitions}, the cutting points $o_5$, $o_{12}$, $o_{14}$, and $o_{19}$ divide the entire sequence $W_1$ into five partitions.

\begin{algorithm}[t]
\footnotesize
  \caption{Partition Construction}
  \label{algo: MSP algorithm}
  \KwIn{The object array $o[1,\cdots, n]$}
  \KwOut{$M$ values; partitions $\mathcal{P}$; indices of the largest $M$ in each partition $\mathcal{LM}$, and set of cutting points $\mathcal{C}$.}
  Initialize array $M[1,\cdots, n]$, partition set $\mathcal{P}$,
   the largest $M$ index set $\mathcal{LM}$, {the cutting point index set $\mathcal{C}$,}
  the right cutting point index $rc$\;
  $M[n] \gets o[n]$, 
  $rc\gets n+1$, $LM\gets n$\;
  \For {$i\gets n-1$ to $1$}{
    \If (\tcp*[h]{scan a cutting point}) {$M[i+1]\leq 0$}{

        {$\mathcal{C}$.\texttt{add} ($i+1$)}; 
        $M[i] \gets o[i]$  \;
        
        \If  {$P[i+1,{rc}-1]$ is a non-degenerate partition}{
             $\mathcal{P}$.\texttt{add} ($P[i+1,{rc}-1]$), $\mathcal{LM}$.\texttt{add} ($LM$)\;
        }
        $rc\gets i+1$, $LM\gets i$\;
    }
    \Else (\tcp*[h]{enter a non-degenerate partition}) {
        $M[i] \gets o[i]+M[i+1]$\;
        \If{$M[i]>M[LM]$}{
            $LM\gets i$\;
        }
       
    }
    
  }
  
    \If{$P[1,{rc}-1]$ is a non-degenerate partition}{
            $\mathcal{P}$.\texttt{add} ($P[1,{rc}-1]$), $\mathcal{LM}$.\texttt{add} ($LM$)\;
            }
    
  \Return{$M$, $\mathcal{P}$, $\mathcal{LM}$, $\mathcal{C}$}\;
\end{algorithm}

The partition construction process is detailed in Algorithm~\ref{algo: MSP algorithm}. After initializing the key data structures (Lines 1-2), the algorithm performs a reverse scan of the window $W$ to compute $M[i]$ values and identify cutting points (Lines 3-12). It maintains a pointer, $rc$, to track the position of previous cutting point. When a new cutting point is found at position $i+1$ (Lines 4-5), objects within the segment $[i+1, rc-1]$ are finalized as a new partition (Lines 6-7). In addition, $rc$ is updated to $i+1$ (Line 8). Here, we want to highlight that not all partitions are retained in our algorithm; a degenerate partition containing only a single non-positive object is discarded, as it contain zero positive-sum maximal intervals. Only non-degenerate partitions are returned to serve \query\ query.

The algorithm also maintains a list $\mathcal{LM}$ that stores, for each non-degenerate partition, the indices of the largest $M[i]$ value (Lines 11-12). Formally, for a non-degenerate partition $P_i=[l_i, r_i]$, $LM_i=\arg \max\nolimits_{j\in [l_i,r_i]}M[j]$. Accordingly, the interval $I[LM_i,r_i]$ is the top-1 maximal interval in $P_i$ and meanwhile also a valid maximal interval in $W$, as formally established in Lemma~\ref{lemma:top1in p}. The proof is provided in our technical report~\cite{long-version} due to space constraints. 
List $\mathcal{LM}$ provides direct access to the top-1 maximal interval in each partition, facilitating the processing of \query\ as detailed later. 

\begin{lemma} \label{lemma:top1in p}
    For a non-degenerate partition $P_j$ that ends at position $r_j$, the interval $I[LM_j,r_j]$ is the maximal interval with the largest sum within $P_j$ and also a maximal interval globally. 
\end{lemma}

In addition, the algorithm uses a list $\mathcal{C}$ to store the indices of all cutting points. These cutting points help determine how partitions evolve
after the query window slides, as detailed in Section~\ref{sec:Update Partitions}. Finally, we emphasize that the partitions generated by this process adhere strictly to the definition given in Definition~\ref{defn:partition}. Specifically, as established by Lemma~\ref{lemma:I_i}, the right boundary of each partition $P=[l_i, r_i]$ is correctly positioned at $r_i$, i.e., $\forall j\in[l_i, r_i]$, $E[j]=r_i$. The proof is omitted, as it can be derived directly from the definition of cutting points. This property is precisely why Algorithm~\ref{algo: MSP algorithm} does not need to maintain $E[\cdot]$ array explicitly.

\begin{lemma}\label{lemma:I_i}
In a non-degenerate partition $P_i=[l_i, r_i]$ constructed by Algorithm~\ref{algo: MSP algorithm}, $r_i$ is $P_i$'s right boundary, i.e., $\forall j\in[l_i,r_i]$, $E[j]=r_i$. 
\end{lemma}

The example in Figure~\ref{fig:constructing partitions} demonstrates this process. The right boundary $rc$ is initialized to $20$, making position 19 the initial right boundary of the initial scanned partition. Since $M[19]<0$, the segment $[19,20)$ is identified as a degenerate partition (containing only $o_{19}$) and is discarded. 
The algorithm then processes a non-degenerate partition. Upon scanning the next cutting point at position 14, a non-degenerate partition $P_4=[14,19)$ is recorded. $LM_4 = 17$, as $M[17]$ is the maximum in this partition. Partitions $P_3$ and $P_2$ are identified similarly. Finally, the leftmost partition $P_1=[1,4]$ is verified as non-degenerate and recorded, with $LM_1=3$.

\subsection{Partition-based $k$-Maximal Interval Search}\label{sec:Partition-based search}

In a streaming context, the query window evolves continuously. A naive method would compute all maximal intervals for each new window, which is computationally expensive. Our key observation is that the top-$k$ results in any window can originate from at most $k$ partitions. This follows from the fact that each maximal interval lies entirely within a single partition. Since the top-$k$ maximal intervals must be mutually disjoint, these $k$ intervals can belong to at most $k$ partitions. 
This property enables effective pruning of the search space by restricting attention to a limited subset of promising partitions.
We therefore rely on the top-1 maximal interval within each partition, identified during partition construction, to determine the $k$ most promising partitions.

Algorithm~\ref{algo: topk search} presents the ``Partition-based $k$-Maximal Interval Search'' algorithm. It maintains a size-$k$ AVL tree that stores the sums of the current top-$k$ maximal intervals discovered so far. 

Initially, the sums of top-1 intervals of all partitions are inserted into the AVL tree (Lines 1-2). Those partitions $P_j$ whose top-1 intervals (e.g., the interval $I[LM_j,r_j]$) remain within this tree are considered promising. 
The algorithm then processes these promising partitions (up to $k$) in descending order of their top-1 interval's sum (Lines 3-6). For each such partition $P_j$, an evaluation is invoked to locate the remaining valid maximal intervals within the partition that are disjoint from the top-1 interval $I[LM_j,r_j]$ (Line 4). The details of this evaluation procedure are described in the next section. The maximal intervals discovered in $P_j$ are stored in a set ${CM}_j$. The sums of these newly discovered maximal intervals are inserted into the AVL tree (Lines 5-6), which continuously maintains the sums of current top-$k$ maximal intervals found so far. The algorithm terminates once all partitions whose top-1 interval sums remain among the current top-$k$ candidates have been evaluated.
Finally, the union of all such ${CM}_j$ forms the final candidate maximal interval set $\mathcal{CM}$, which will be used for subsequent window updates, and the maximal intervals whose sums remain in the final AVL tree form the final top-$k$ results (Line 7).

The results returned by Algorithm~\ref{algo: topk search} are guaranteed to be the top-$k$ maximal intervals, as stated in Lemma~\ref{lemma:exact answer}. 
We omit the proof for brevity and provide it in Lemma~\ref{lemma:exact answer} of our technical report~\cite{long-version}.

\begin{lemma}\label{lemma:exact answer}
Within a query window $W_i$, the $k$  intervals returned by Algorithm~\ref{algo: topk search} correspond to the top-$k$ maximal intervals in the window.
\end{lemma}

\begin{algorithm}[t]
\footnotesize
  \caption{Partition-based $k$-Maximal Interval Search 
  }
  \label{algo: topk search}
  \KwIn{Partitions $\mathcal{P}=\{P_1,P_2\cdots,P_{p}\}$ with their respective top-1 $M$ values $(M[LM_1], M[LM_2],\cdots,M[LM_p] )$, the candidate maximal interval set $\mathcal{CM}=\{{CM}_1,{CM}_2,\cdots,{CM}_p\}$, required $k$
  }
  \KwOut{The top-$k$ maximal intervals, the candidate maximal interval set.}

  Initialize an AVL tree $avl$ with the size of $k$\;
    insert all $M[LM_j]$ for $j\in[1,p]$ into $avl$\;
  \While{find $M[LM_i]$ in $avl$ in descending order}{
    incrementally update $\mathcal{CM}_i$; \tcp*[h]{Algorithm~\ref{algo: incremental maintenance Maximal Intervals}}
    
    \ForEach{maximal interval $MaxI[u,v]$ in ${CM}_i$}{
        insert the sum of $MaxI[u,v]$ into $avl$\;}
  }
  \Return{maximal intervals whose sums are in $avl$, $\mathcal{CM}$}\;
\end{algorithm}

\subsubsection{Identifying maximal intervals in a partition} \label{sec:no-max maximal in part}

Algorithm~\ref{algo: topk search} invokes an evaluation procedure to identify all remaining maximal intervals within each promising partition $P_j$; specifically, those maximal intervals disjoint from the top-1 interval $I[LM_j,r_j]$. 
To support this procedure, we first derive a necessary but not sufficient condition that any maximal interval must satisfy, as formalized in Lemma~\ref{lemma:head-tail}. Again, we omit the proof due to space constraints; the full proof is provided in Lemma~\ref{lemma:head-tail} of our technical report~\cite{long-version}.

\begin{lemma} \label{lemma:head-tail}
    Let $MaxI[u,v-1]$ be a maximal interval in the sub-range $[l_j,LM_j-1]$ of partition $P_j=[l_j, r_j]$. Then,  $M[u]> M[i]$ for all $i\in (u, v]$, and $M[v]<M[i]$ for all $i\in [u,v)$.
    Thus, $u$ is a strict local maximum and $v$ is a strict local minimum of the $M[\cdot]$ values. 
\end{lemma}

Lemma~\ref{lemma:head-tail} states that $I[u, v-1]$ can be a maximal interval only if the corresponding $M$-subsequence $M[u,v]$ has $M[u]$ as its strict maximum and $M[v]$ as its strict minimum. Any $M$-subsequence satisfying this condition is defined as a \emph{Potential Target $M$-Subsequence}
(pTMS for short). All other $M$-subsequences can therefore be safely filtered out. However, not every pTMS corresponds to a maximal interval. For instance, in Figure~\ref{fig:non-maximum}, $M[5,7]$ is a pTMS, yet $I[5,6]$ is not a maximal interval. Thus, while pTMS provides a filtered candidate set, additional refinement is needed. 

We therefore seek a \emph{Target $M$-Subsequence} (TMS for short), i.e., an $M[u,v]$ whose corresponding interval $I[u,v-1]$ is indeed a maximal interval. The key observation underlying this refinement is the mergeability property of pTMS. Two pTMS, $M[u_1,v_1]$ and $M[u_2,v_2]$ (with $v_1<u_2$) are mergeable if their concatenation $M[u_1,v_2]$ also constitutes a  pTMS. Based on this notion, we establish a key result (Lemma~\ref{lemma:non-mergeable}):  a pTMS is a true TMS if and only if it is non-mergeable with any other pTMS in the sequence. Again, we omit the proof for brevity and provide it in Lemma~\ref{lemma:non-mergeable} of our technical report~\cite{long-version}.

\begin{figure}[t]
    \centering
            \includegraphics[width=\linewidth]{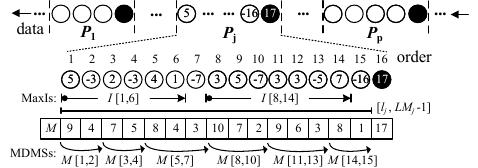}   
            \caption{An example of maximal intervals in range $[l_j, LM_j-1]$ in $P_j$ where $l_j=1$ and $LM_j=16$.}
\label{fig:non-maximum}
\end{figure}

\begin{definition}\label{defn:mergeable}
    \textbf{Mergeability of two pTMS.} Let $M[u_1,v_1]$ and $M[u_2, v_2]$ be two pTMS in a partition $P_j=[l_j,r_j]$ such that $l_j\le u_1<v_1<u_2<v_2< LM_j$. They are mergeable iff the concatenated sequence $M[u_1, v_2]$ itself constitutes a valid pTMS (i.e., $M[u_1]$ is the strict maximum and $M[v_2]$ is the strict minimum over the extended range $[u_1,v_2]$). Otherwise, the two pTMS are non-mergeable. 
\end{definition}

\begin{lemma} \label{lemma:non-mergeable}
    A pTMS within $[l_j,LM_j-1]$ is a TMS iff it is non-mergeable with any other pTMS in $[l_j,LM_j-1]$.
\end{lemma}

Guided by Lemma~\ref{lemma:non-mergeable}, our task can be reduced to identifying non-mergeable pTMS. 
Directly determining non-mergeability of a pTMS $M[u_1,v_1]$ during a scan of a $M$ sequence requires examining all future $M$ values, leading to repeated scans. 
To avoid this inefficiency, 
we introduce \emph{Maximal Monotonically Decreasing $M$-SubSequence} (MDMS for short), the longest contiguous subsequences where $M$ values directly decrease (e.g., $M[1, 2]$ and $M[3, 4]$ in Figure~\ref{fig:non-maximum}). MDMS plays a crucial role in our evaluation procedure because it has three useful properties: i) each MDMS is itself a pTMS; ii) all MDMS can be identified in a single scan of the sequence, and iii) MDMS can be iteratively merged according to Definition~\ref{defn:mergeable} to eventually produce the desired TMS.

\begin{algorithm}[t]
\footnotesize
  \caption{Identify Maximal Intervals
  }
  \label{algo:non-maximum}
  \KwIn{Array $M[l_a \dots LM_a-1]$}
  \KwOut{Maximal intervals}
  Initialize an empty double linked list $L$ with a dummy head\;
    
  \While{traverse the array $M$ from $M[l_a]$ to $M[LM_a-1]$}{
    scan a Maximal Monotonically Decreasing $M$ SubSequence $M[i, j]$\;
    
    $b\gets |L|-1$; $L.append(M[i,j])$\;
    \While (\tcp*[h]{traverse $L$ for mergeability checks}) {$b\ge 0$ }
    {
        $M[u,v]\gets L[b]$\;
    	\If{$M[u]>M[i]\land M[v]>M[j]$
        \tcp*[h]{mergeable}}
    	{
            $L[b]\gets M[u,j]$; $L[b].next\gets NULL$\;
            $b\gets b-1$; $M[i,j]\gets M[u,j]$\;
            }
        \ElseIf{$M[u]\le M[i]\land M[v]> M[j]$}
        {
            $L.M[i,j].prev \gets L[b].prev$; $b\gets b-1$\;
        }
        \Else(\tcp*[h]{$M[v] \le M[j]$})
        {
            $L.M[i,j].prev \gets L[b]$; 
            break; \tcp*[h]{stop traversing $L$}\\
        }
    }
    }
  
  \Return elements of $L$\;
\end{algorithm}

Based on these observations, we present the procedure to identify maximal intervals in a given $M$ sequence, as outlined in Algorithm~\ref{algo:non-maximum}. The algorithm scans the $M$ sequence in each partition $P_a$ from $l_a$ to $LM_a-1$, identifying MDMS on the fly (Lines 2-13). During the scan, we maintain an ordered list $L$ that stores current pTMS candidates. Whenever a new MDMS $M[i,j]$ is identified, it is appended to the tail of $L$ (Lines 3-4). 
We perform mergeability checks to examine whether it can merge with existing pTMS in $L$, examining them sequentially from the tail (the element in $L$ immediately preceding $M[i,j]$) toward the head (Lines 5-13). If a merge is possible, the two subsequences are merged to form a longer pTMS; the new and longer pTMS becomes the new tail so all pTMS contained within the merged range are removed from $L$. We then continue checking whether the newly formed pTMS can be further merged with earlier elements in $L$. 
Once this local mergeability checks complete, the tail pTMS (originating from the identified MDMS) is guaranteed to be non-mergeable with any other pTMS in $L$. The algorithm then proceeds to identify the next MDMS and repeats the mergeability checks. 
After all MDMS are processed, the final list $L$ contains the desired TMSs, each corresponding to a maximal interval.

\begin{figure}[t]
    \centering
            \includegraphics[width=\linewidth]{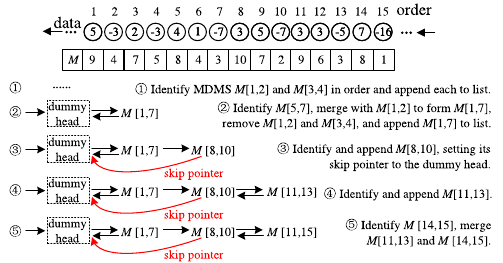}   
            \caption{Build a skip-pointer when the algorithm scans $M[8, 10]$.
            }
    \label{fig:pointer}
\end{figure}

Figure~\ref{fig:pointer} shows this process when MDMS $M[5,7]$ is scanned (Step \textcircled{2}). The list $L$ contains $M[1, 2]$ and $M[3, 4]$. The algorithm checks its mergeability with these candidates from tail to head sequentially.
While $M[5, 7]$ cannot merge with $M[3, 4]$, it is mergeable with $M[1, 2]$, forming the new, longer pTMS $M[1, 7]$. This new pTMS now fully encompasses the range of $M[3, 4]$, making it obsolete. Consequently, the list is updated to $M[1, 7]$.

A naive implementation of mergeability checks that examines each new MDMS against every pTMS in $L$ is inefficient, as it performs many unnecessary checks. We optimize it with a \emph{skip-pointer} mechanism, based on two key insights: let $M[u,v]$ and $M[i,j]$ be two pTMSs in $L$ with $v<i$,
1) if $M[v] \le M[j]$, $M[i, j]$ \emph{cannot} merge with $M[u,v]$ and any earlier pTMS in $L$, as the merged sequence would violate the requirement that its last element be the strict minimum $M$ value; 2) if $M[u] \le M[i]$, $M[u, v]$ \emph{cannot} merge with $M[i,j]$ or any subsequent pTMSs in $L$, since the merged sequence would violate the requirement that the first element of a pTMS must have the strict maximum $M$ value over the merged range.

Consequently, the mergeability of a pTMS (a new MDMS $M[i,j]$ or one formed by a previous merge) and an existing pTMS $M[u,v]$ in $L$ is determined by evaluating the extended range $M[u,j]$. This leads to three cases based on whether $M[u]$ is the strict maximum and $M[j]$ is the strict minimum in the extended range $M[u,j]$. 

\begin{itemize}[leftmargin=*]
\item{

Case i): both conditions hold (Lines 7-9). This sequence is merged into a new pTMS $M[u,j]$. All pTMS contained within this new range are removed from $L$, and the algorithm continues checking the new pTMS against the remaining elements in $L$.
}

\item{

Case ii): only $M[j]$ is minimal (Lines 10-11). $M[u,v]$ can not merge with $M[i,j]$ and future pTMS located after $M[i,j]$. However, $M[i,j]$ or future pTMSs $M[i',j']$ with $i'>j$ may be able to merge with pTMS $M[u',v']$ located before $M[u,v]$ in $L$ if $M[u']$ has the strict maximum $M$ value while $M[j]$ or $M[j']$ has the strict minimum $M$ value in the range $[u', j]$ or $[u', j']$. Consequently, only $M[u,v]$ can be bypassed in this mergeability evaluation procedure, and $\texttt{prev}$ pointer of $M[i,j]$ is set to that of $M[u,v]$, and the algorithm continues checking $M[i,j]$ against the preceding element in $L$.
}

\item{

Case iii): $M[j]$ is not minimal (Lines 12-13). $M[i,j]$ cannot merge with $M[u,v]$ or any pTMS located before $M[u,v]$ in $L$. Consequently, we can safely terminate the mergeability checks of $M[i,j]$. In addition, $\texttt{prev}$ pointer of $M[i,j]$ is set to $M[u,v]$.
}

\end{itemize}

In summary, we implement both insights in above mergeability check procedure, and we re-purpose the $\texttt{prev}$ pointer in the ordered list $L$. For any pTMS $M[i,j]$ in $L$, its $\texttt{prev}$ pointer does not simply point to the immediate predecessor in $L$. Instead, it points to the nearest preceding pTMS $M[u,v]$ in $L$ such that $M[v]\le M[j]$, or to the dummy head element if no such preceding pTMS exists. This pointer effectively creates a ``skip list'' that allows the algorithm to bypass directly all intervening pTMS that are guaranteed to be unmergeable due to the fact that their initial $M$ value is not the strict maximum and their last $M$ value is not the strict minimum.

For example, as shown in Figure~\ref{fig:pointer}, when the scanning MDMS $M[8,10]$ against a list containing $M[1, 7]$, its \texttt{prev} pointer links to the dummy head. This instantly tells that $M[1, 7]$ cannot merge with any subsequent pTMS, due to the existence of $M[8,10]$.
When MDMS $M[11, 13]$ is formed, we only need to check whether it is mergeable with $M[8,10]$, bypassing the check against $M[1,7]$.

\noindent
\textbf{Complexity analysis}. Algorithm~\ref{algo: MSP algorithm} scans the $n$ objects in the query window to construct $|\mathcal{P}|$ partitions (where $|\mathcal{P}|<n$). This costs $\mathcal{O}(n)$ time and $\mathcal{O}(n)$ space. Algorithm~\ref{algo: topk search} first spends $\mathcal{O}(|\mathcal{P}|\cdot log~k)$ time to find top-$k$ promising partitions. It then traverses the AVL tree, which costs $\mathcal{O}(k \cdot log~k)$ time. During this process,  Algorithm~\ref{algo:non-maximum} is called to identify maximal intervals within each partition. Assuming partitions are distributed uniformly, each partition $P_i$ contains $\tfrac{n}{|\mathcal{P}|}$ objects and $|CM_i|$ 
maximal intervals. Algorithm~\ref{algo:non-maximum} is linear, so identifying them costs $\mathcal{O}(\tfrac{n}{|\mathcal{P}|})$ time. They participate in comparison, requiring $\mathcal{O}(|CM_i|\cdot log~k)$ time. The total time complexity is $\mathcal{O}\big(n + |p|\cdot log~k + k\cdot log~k + k \cdot (\tfrac{n}{|\mathcal{P}|}+|CM_i|\cdot log~k)\big)$ $= \mathcal{O}\big((1+\tfrac{k}{|\mathcal{P}|})\cdot n + (|P|+k+|CM|)\cdot log~k\big)$. 

For space complexity, Algorithm~\ref{algo: MSP algorithm} maintains $n$ $M$ values, $|\mathcal{P}|$ partitions including their boundaries, $\mathcal{LM}$ set, and cutting points. Algorithm~\ref{algo: topk search} maintains an AVL tree of size $k$ and the candidate set $\mathcal{CM}$. Algorithm~\ref{algo:non-maximum} processes a partition requiring space bounded by $\mathcal{O}(\tfrac{n}{|\mathcal{P}|})$; when it completes a partition, the space can be released. Thus, the total space is bounded by $\mathcal{O}(n + |\mathcal{P}| + k + |\mathcal{CM}| + \tfrac{n}{|\mathcal{P}|})$.

\section{The Incremental Maintenance} \label{sec:inc-maint}

Our partitioning strategy is designed not only to support \query\ search within a single window but also to incrementally update the $\mathcal{LM}$ and $\mathcal{CM}$ sets\footnote{$LM_j\in \mathcal{LM}$ and $\mathcal{CM}_i\in \mathcal{CM}$ record the position of the top-1 interval and the remaining maximal intervals, if identified, in partition $P_j$, respectively.}, which are essential for identifying promising partitions and searching for results in the new window. In the following, we first explain how to update the partitions and then detail the maintenance of the $\mathcal{LM}$ and $\mathcal{CM}$ sets corresponding to the updated window.

\subsection{Update Partitions} \label{sec:Update Partitions}

Sliding the window can significantly alter partitions. As shown in Figure~\ref{fig:back propagation}, 
when $W_1$ slides to $W_2$,  partitions evolve: $P_2$ shrinks,  $P_4$ extends,
and $P_4'$ emerges.
These changes result from updated cutting points, which shift from \{$o_5$, $o_{12}$, $o_{14}$, $o_{19}$\} to \{$o_{12}$, $o_{14}$, $o_{22}$\}. A naive approach would track these changes by recomputing all affected $M$ values via a full window scan. However, this is highly inefficient, as many objects' $M$ values remain unchanged and many updated partitions do not contribute to query results. To eliminate these unnecessary computations, we propose a maintenance strategy that updates only a subset of the $M$ values.

Each time the window $W$ advances, the first $s$ objects expire and $s$ new objects arrive. We first calculate the $M$ values for the new objects, in reverse order, from $M[n+s]$ to $M[n+1]$, using Equation~\eqref{equ:dp}. If $\Delta=M[n+1]\le 0$, the update is complete, as this non-positive value blocks further propagation to preceding $M$ values. Otherwise, we propagate the update backward through the existing partitions. This backward propagation process traverses the partitions in reverse order, starting from the last partition $P_p$ in $W_1$. For each partition $P_i=[l_i,r_i]$, if $\Delta_i=M[r_{i}+1]>0$, we add  $\Delta_i$ to all $M$ values within $P_i$. We then proceed to partition $P_{i-1}$. The propagation terminates at the first partition $P_j$ where $\Delta_j=M[r_j+1]\le 0$, as this cutting point prevents the cumulative impact from propagating further. Note that all partitions, including degenerate partitions (which can be easily recovered from cutting points), participate in this process. 

\begin{figure}[t]
    \centering
            \includegraphics[width=\linewidth]{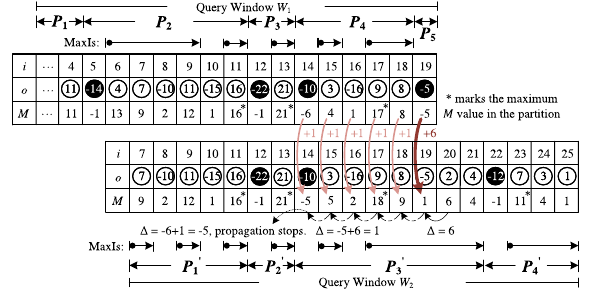}   
            \caption{An example of window sliding from $W_1$ to $W_2$.
            }
    \label{fig:back propagation}
\end{figure}

Figure~\ref{fig:back propagation} illustrates this process. When $W_1$ advances to $W_2$ by $s=6$ objects, we first calculate the $M$ values for the new objects, i.e., $M[20,25]$. Since $\Delta=M[20]=6>0$, we propagate the update backward through the partitions: 1) Partition $P_5$: the $M$ value is increased by $\Delta=6$. 2) Partition $P_4$: As $\Delta_4=M[19]=1>0$, propagation continues and $M$ values are increased by $\Delta_4$. 
3) Partition $P_3$: Since $\Delta_3=M[14]=-5$, the propagation terminates, leaving $M$ values for objects in partitions $P_3$, $P_2$, and $P_1$ unchanged. 

We want to highlight that when propagating the update backward through the existing partitions, $\Delta$ value reduces monotonically, as stated in Lemma~\ref{lemma:delta-decrease}. The proof is skipped, as it can be derived by the fact that $\Delta_i=\Delta_{i+1}+M[r_i+1]$ and $M[r_i+1]=M[l_{i+1}]<0$ (position $l_{i+1}$ is a cutting point). In addition, a non-cutting point in the original window will not become a cutting point in the updated window, as the $M$ values of objects in a partition $P_i$ change only when the corresponding $\Delta_i>0$. A non-cutting point has a positive $M$ value, and this value can only increase (or remain unchanged) during backward propagation. Consequently, it cannot become a cutting point, as stated in Lemma~\ref{lemma:cuttingpoint}. 

\begin{lemma}\label{lemma:delta-decrease}
Let $P_i$ and $P_{i+1}$ be consecutive partitions whose $M$ values are updated by backward propagation with changes $\Delta_{i}$ and $\Delta_{i+1}$ respectively. Then, $\Delta_{i}<\Delta_{i+1}$.
\end{lemma}

\begin{lemma} \label{lemma:cuttingpoint}
When the window advances from $W$ to $W'$, let $\mathcal{C}$ and $\mathcal{C}'$ indicate the sets of cutting points in $W$ and $W'$ respectively, and let $\mathcal{E}$ and $\mathcal{A}$ denote the sets of expired and newly arrived objects. A non-cutting point in $W$ will certainly not become a cutting point in $\mathcal{C}'$. Moreover, $\mathcal{C}'$ is a subset of $(\mathcal{C}\setminus \mathcal{E})\cup \mathcal{A}$.
\end{lemma}

After the $M$ values are updated, cutting points can be identified to form partitions for the new window. In our example shown in Figure~\ref{fig:back propagation}, given cutting points $o_{12}$, $o_{14}$, and $o_{22}$, four partitions are formed: $P'_1=[7,11]$, $P'_2=[12,13]$, $P'_3=[14,21]$, and $P'_4=[22,25]$.   
We observe that partition evolution can be described using four fundamental patterns, illustrated in Figure~\ref{fig:back propagation}: 
(\textit{i}) Shrinking, where a partition loses objects from its head without incorporate new ones (e.g., $P_1'$); (\textit{ii}) Unaltered, where a partition remains unaltered (e.g., $P_2'$); (\textit{iii}) Extended, where a partition extends by shifting its right boundary forward (e.g., $P_3'$); or (\textit{iv}) Emergent, where a new partition forms entirely from newly arrived objects (e.g., $P_4'$). 
If a shrinking partition exists, it must be the first partition in the new window. For presentation clarity, we assume that all partition changes conform to one of these four patterns. This categorization is well supported by the property that a non-cutting point cannot become a cutting point after a window slides (Lemma~\ref{lemma:cuttingpoint}). However, in an extreme case, the first partition in the new window may both lose initial objects (due to object expiry) and extend its range, although this occurs only if backward propagation reaches the first cutting point, which is not common in practice.
Importantly, our approach fully supports this corner case. It can be interpreted as a combination of a Shrinking followed by an Extension, thus does not affect the correctness of the algorithm or the returned top-$k$ intervals.

\begin{algorithm}[t]
\footnotesize
  \caption{Partition Update}
  \label{algo: incremental maintenance partitions}
  \KwIn{Cutting points $\mathcal{C}$, expired objects $\mathcal{E}$, newly arrived objects $\mathcal{A}$, $M$ array, $\mathcal{LM}$ list}
  \KwOut{Updated cutting points $\mathcal{C}'$, and updated $\mathcal{LM}'$ list}

    $\mathcal{C}' \gets \mathcal{C} \setminus \mathcal{E}$\;
    $M_{new}, \mathcal{C}_{new}  \gets$ \texttt{DP-BuildNewCPs}($\mathcal{A}$)
    \tcp*{call Algorithm~\ref{algo: MSP algorithm}}
    $\mathcal{C}' \gets \mathcal{C}' \cup \mathcal{C}_{new}\setminus \texttt{BackwardPropagation}(\mathcal{C}', M, M_{new}$)\;
    $\mathcal{LM}'\gets$ \texttt{UpdateLM}($\mathcal{C}'$, $\mathcal{C}$, $\mathcal{LM}$, $M$, $M_{new}$)\;
  \Return{$\mathcal{C}'$, $\mathcal{LM}'$}\;
\end{algorithm}
 
Algorithm~\ref{algo: incremental maintenance partitions} updates partitions as follows. It removes expired cutting points (Line 1), processes new objects via Algorithm~\ref{algo: MSP algorithm} to find new cutting points (Line 2), and runs \texttt{BackwardPropagation} function propagates their cumulative impact backward to remove cutting points whose updated $M$ values become positive (Line 3). 
Finally, function \texttt{UpdateLM} updates $LM'$, the index of top-1 $M$ value in each partition $P'$ in the new window (Line 4). Let $P'=[l',r']$ be an updated partition in the new window. We detail its $LM'$ update operation below, dependent on the evolution pattern of $P'$. i) Shrinking (from partition $P_i$): If $l'\le LM_i$, $LM'=LM_i$; otherwise, we scan $M$ values of objects in $P'$ to locate $LM'$. 
ii) Unaltered (same as partition $P_i$ in previous window): $LM'=LM_i$.
iii) Extended: assume $P'$ covers $P_i$, $\cdots$, $P_j$ ($j\ge i$) in the original window, and possibly including new objects $\mathcal{A}'\subseteq \mathcal{A}$. Due to the nature of backward propagation, the position $LM_i$ within any partition $P_i$ in $W$ continues to have the largest $M$ value among all objects in that partition, even after their $M$ values are updated. Consequently, $LM'$ must lie among the positions \{$LM_i$, $\cdots$, $LM_j$\}, and, when $\mathcal{A}'\ne \emptyset$, among the newly arrived objects \{$n+1$, $\cdots$, $n+|\mathcal{A}'|$ \}, where $n$ indicates the right boundary of the previous window $W$. We identify $LM'$ by scanning the $M$ values at these positions. 
iv) Emergent: we scan $M$ values of objects in $P'$ to locate $LM'$.

\noindent
\textbf{Implementation Optimization.} In the above description, we assume that the backward propagation stops only after encountering a cutting point say $c_i$ with non-positive $M$ value. This is mainly for ease of explanation. 
In practice, however, it is not necessary to update the $M$ values for all objects located after $c_i$. In our implementation, we adopt a lazy update strategy. We update the $M$ values only for new objects and for  existing cutting points associated with previous window. The $M$ values of objects within a partition are updated only when necessary (e.g., when determining $LM'$ for a shrinking partition). A binary flag associated with each partition records whether the $M$ values of its objects are updated. When a window slides, the flags of all the partitions are initialized to 0.

\noindent
\textbf{Complexity analysis.} Removing expired cutting points requires $O(|\mathcal{E}|)$ time. 
Constructing new cutting points from the newly arrived objects costs $O(|\mathcal{A}|)$ time. 
Backward impact propagation only updates cutting points' $M$ values rather than all, which requires $O(|\mathcal{C}|)$. 
The total time complexity is $O(|\mathcal{E}| + |\mathcal{A}| + |\mathcal{C}|)$. We need auxiliary space to record the cutting points and the $M$ values of new objects, so the total space complexity is  $O(|\mathcal{C}| + |\mathcal{C}_{new}| + |\mathcal{A}|)$.

\begin{algorithm}[t]
\footnotesize
  \caption{Incremental Maximal Interval Maintenance}
  \label{algo: incremental maintenance Maximal Intervals}
  \KwIn{Partition $P'=[l',r']$ in new window with $LM'$ indicating the top-1 $M$ value, original partitions $\mathcal{P}=\{P_1$, $\cdots$, $P_p\}$, original $\mathcal{LM}$ list, reusable maximal intervals $\mathcal{CM}=$\{$CM_1$, $\cdots$, ${CM}_p\}$}
  \KwOut{Maximal intervals ${CM}'$ in partition $P'$}
  \If{$\exists P_j\in \mathcal{P}$, \texttt{Unaltered}($P'$, $P_j$) $\land$ \texttt{MaxIReusable}($P_j$)}
  {
    ${CM}'\gets{CM}_j$\;
  }
  \ElseIf{$\exists$ $P_j\in \mathcal{P}$, \texttt{Shrinking}($P'$, $P_j$) $\land$ \texttt{MaxIReusable}($P_j$)}
  {
    ${CM}' \gets {CM}_j\setminus \{\text{expired } MaxI[a,b]$ with $b<l'$\} \;
    \If{$\exists$ $MaxI[a,b]\in {CM}'$ such that $l'\in(a,b]$}
        {
        ${CM}'.\texttt{remove}(MaxI[a,b])$\;
        ${CM}'.$\texttt{add}(reidentify $MaxI$ within $[l', b]$ via Algorithm~\ref{algo:non-maximum})\;
        }
  }
  \ElseIf{$\exists P_i, \cdots, P_j\in \mathcal{P}$, \texttt{Extended}($P'$, $P_i$, $\cdots$, $P_j$) }
    {
            \If(\tcp*[h]{Case 1}){$\exists  a\in[i,j]$, $LM_a=LM'$ $\land$ 
            \texttt{MaxIReusable}($P_i$, $\cdots$, $P_a$)}{
                ${CM}' \gets \cup_{t=i}^a {CM}_t \cup MaxI[LM',r']\setminus \{MaxI \text{ in } [LM',r']$\}\;
            }
            \ElseIf(\tcp*[h]{Case 2}){$LM'> n$ $\land$ \texttt{MaxIReusable}($P_i$, $\cdots$, $P_j$)}{
                ${CM}'\gets MaxI[LM',r'] \cup$ reidentify $MaxI$ within $[l',LM'-1]$ via Algorithm~\ref{algo:non-maximum}\;
                
            }
        } 
    \Else(\tcp*[h]{emergent partition or partitions without reusable MaxI}){
        ${CM}'\gets$ Identify Maximal Interval via Algorithm~\ref{algo:non-maximum}\;
    }
  \KwRet{${CM}'$}\;
\end{algorithm}

\subsection{Maintain Maximal Intervals Incrementally}\label{sec:Maintaining Maximal Intervals}

After updating the partitions and $\mathcal{LM}$, we answer \query\ using the partition-based $k$-maximal interval search (Algorithm~\ref{algo: topk search}). Recall that $\mathcal{LM}$ only provides access to the top-1 maximal interval in each partition. During AVL-tree traversal, however, we must still retrieve other maximal intervals (beyond the top-1) from each promising partition whose top-1 enters the current top-$k$. A naive strategy computes all maximal intervals in these partitions using Algorithm~\ref{algo:non-maximum}, but this is inefficient because it ignores maximal intervals already obtained in the previous window (maintained in set $\mathcal{CM}$). Depending on how partitions evolve, many intervals in $\mathcal{CM}$ remain valid and can be reused. To avoid redundant computation, we incrementally maintain maximal intervals by exploiting partition evolution patterns and reusing prior results. Algorithm~\ref{algo: incremental maintenance Maximal Intervals} summarizes this procedure: it first determines the current partition's evolution pattern and the availability of reusable maximal intervals, then applies the corresponding maintenance strategy.

Specifically, for new partitions or partitions without reusable maximal intervals (Lines 13-14), maximal intervals must be computed from scratch using Algorithm~\ref{algo:non-maximum}. 
When reusable maximal intervals exist (determined by function \texttt{MaxIReusable}), unaltered partitions can directly reuse all previously computed maximal intervals (Lines 1-2). For shrinking and extended partitions, the corresponding maintenance strategies are detailed below. 

\noindent
\textbf{Shrinking Partitions}.  
Let partition $P'=[l',r]$ be a shrinking partition derived from partition $P=[l,r]$, where $l'>l$. The maximal intervals containing only expired objects are removed. If position $l'-1$ falls within an existing maximal interval $MaxI[a,b]$, we need to perform the following updates. 
First, $MaxI[a,b]$ shrinks to $I[l',b]$. It is guaranteed that pTMS $M[l',b+1]$ corresponding to interval $I[l',b]$ will not merge with any other pTMS in this partition.  
This is because in the maximal interval $MaxI[a,b]$, $M[a]$ has the strict maximal $M$ value, and therefore $M[a]>M[l']$. The fact that $MaxI[a,b]$ is a maximal interval in $P=[l,r]$ and $M[a]>M[l']$ guarantees that $M[l',b+1]$ cannot merge with any other pTMS in this partition.
Next, we invoke Algorithm~\ref{algo:non-maximum} to identify maximal intervals, if any, in the range of $[l',b]$.

In summary, in a shrinking partition, maximal intervals containing only expired objects are removed, the maximal interval containing both expired and valid objects, if it exists, requires further evaluation, 
and all other maximal intervals remain unchanged. 
For example, in Figure~\ref{fig:back propagation}, partition $P_2$ in $W_1$ shrinks to $P'_1$ in $W_2$. Objects $o_5$ and $o_6$ in $P_2$ expire, affecting only maximal interval \(MaxI[6,9]\), which shrinks to interval $I[7, 9]$. We then locate two maximal intervals \(I[7,7]\) and \(I[9,9]\) in this reduced range via Algorithm~\ref{algo:non-maximum}.
The other maximal interval, \(MaxI[11,11]\), remains unchanged.

\noindent
\textbf{Extended Partitions}.
Let partition $P'=[l_i,r']$ be an extended partition originating from $P_i=[l_i,r_i]$ ($r'>r_i$), covering the sequence of partitions $P_i$, $\cdots$, $P_{j}$ for some $j \ge i$, and possibly including new objects $\mathcal{A}'\subseteq \mathcal{A}$.
Let $LM'$ denote the new position in $P'$ with the largest $M$ value. 
The maintenance operations depend on whether $LM'>n$, where $n$ denotes the ending position of the previous window $W$, as described below. 

\noindent
Case 1): $LM'\le n$. This means that the position with the top-1 $M$ value lies in an existing partition, say $P_a$ (i.e., $LM'=LM_a$). Accordingly, $MaxI[LM',r']$ is the top-1 maximal interval in the extended partition $P'$, and all maximal intervals in this range become invalid and are removed. However, maximal intervals that end before position $LM'$ remain unchanged. In other words, if $CM$s corresponding to $P_i$, $\cdots$, $P_a$ are available, we can immediately identify $CM'$ by re-using all maximal intervals ending before $LM_a$. Otherwise, we only need to compute $CM$s of those partitions among $P_i$, $\cdots$, $P_a$ whose $CM$s are unavailable. We did not explicitly highlight this in our pseudo-code due to space constraints and for code simplicity. 
Consider the extended partition $P'_3=[14,21]$ in Figure~\ref{fig:back propagation} as an example. $P'_3$ (with $LM'=17$) originates from $P_4$, covers partitions $P_4$ and $P_5$, and includes new objects $\mathcal{A}'=\{o_{20}, o_{21}\}$. 
Accordingly, $MaxI[17,21]$ becomes the new top-1 maximal interval, and the previous maximal interval $MaxI[17,18]$ is removed. All maximal intervals ending before position 17 (i.e., $MaxI[15,15]$) remain valid. 

\noindent 
Case 2): $LM'>n$. This means the position with the top-1 $M$ value is located among new objects. $I[LM',r']$ becomes the top-1 maximal interval in the extended partition $P'$, we need to relocate the maximal intervals in the range of $[l_i, LM'-1]$ via Algorithm~\ref{algo:non-maximum}.

\noindent
\textbf{Complexity analysis.} Let $|\mathcal{P}|$ partitions with a uniform distribution exist in the window with $n$ objects. 
In Algorithm~\ref{algo: incremental maintenance Maximal Intervals}, only when the worst case (case 2 in extended partitions) occurs, maximal intervals in the partition need to be re-identified. 
In this case, the algorithm degenerates to the scenario where  Algorithm~\ref{algo:non-maximum} processes one partition, which requires $\mathcal{O}(\tfrac{n}{|\mathcal{P}|})$ time and $\mathcal{O}(\tfrac{n}{|\mathcal{P}|})$ space.

\begin{figure*}[tp]
    \centering
            \includegraphics[width=\linewidth]{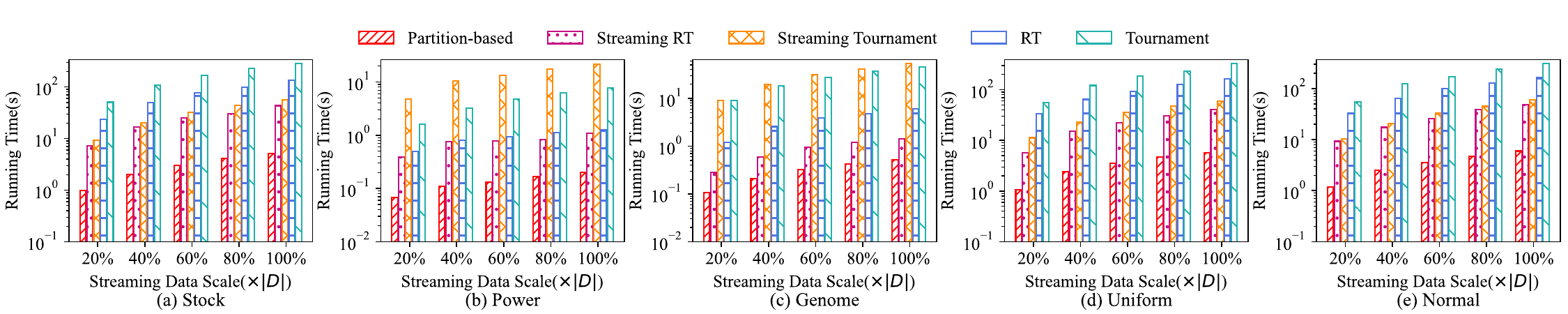}   
            \caption{Running time comparison under different data scales.}
    \label{fig:overall}  
\end{figure*}

\section{Experiments}

\subsection{Experiment Setting} \label{sec:experi setting}

\noindent
\textbf{Datasets.} 
Three real datasets, \textsc{Stock}, \textsc{Power}, and \textsc{Genome}, and two synthetic datasets, \textsc{Uniform} and \textsc{Normal}, are used.  1) \textbf{\textsc{Stock}}. Bitcoin records from Alpha Vantage~\cite{online:Alpha_Vantage} (Nov. 4, 2013--Mar. 31, 2025), converted into a $5$-minute price-delta stream (close price $-$ open price) with $~9$M records; top-$k$ maximal intervals indicate the largest cumulative price increases. 2) \textbf{\textsc{Power}}. UK's half-hourly power data (Nov. 2008--Dec. 2022), compiled by Kaggle from the Elexon Portal and National Grid~\cite{online:Electrical}. Generation minus consumption forms a surplus stream with $~0.24$M records; top-$k$ maximal intervals indicate the largest cumulative surplus. 3) \textbf{\textsc{Genome}}. The Haemophilus inﬂuenzae genome from NCBI~\cite{online:genome}; following prior work~\cite{DBLP:conf/ismb/RuzzoT99}, CpG dinucleotides score $20$ and all others score $-1$, yielding $1.8$M records. 4) \textbf{\textsc{Uniform}} and 5) \textbf{\textsc{Normal}}. Synthetic streams of $10$M real numbers each, drawn from uniform and normal distributions.

\begin{table}[t]
\footnotesize
  \caption{Parameter Setting}
  \label{tab:param}
  \centering
  \begin{tabular}{|c|c|c|}
    \hline
    {\centering Parameters} & {\centering Value Ranges \par} & {\centering Default \par}\\ \hline
    $n$   & $0.1\%$, $0.5\%$, {$1\%$}, $5\%$, $10\%$ $(\times |D|)$ & $1\%\times |D|$\\ \hline
    $s$   & $0.01\%$, {$0.1\%$}, $0.5\%$, $1\%$, $10\%$ $(\times n)$ & $0.1\%\times n$ \\ \hline
    $k$   & $5$, $10$, {$50$}, $100$, $1000$ & $50$ \\ \hline
  \end{tabular}
\end{table}

\noindent
\textbf{Parameter Setting.}
We vary three key parameters: window size $n$, slide size $s$, and result size $k$. For each dataset, we vary one parameter at a time using the values in Table~\ref{tab:param}, while fixing the others at their defaults.
Since suitable values depend on application needs, no single configuration is universally optimal. Following streaming and sliding-window evaluations~\cite{DBLP:journals/tkde/ZhuWYZW17,DBLP:conf/edbt/YangSRW11,DBLP:conf/sigmod/MouratidisBP06,DBLP:conf/cikm/ShastriDRW11,DBLP:journals/pacmmod/ZhuWYZ23,DBLP:journals/pvldb/TranFS16,DBLP:journals/pvldb/YoonLL19,DBLP:journals/tkde/TaoP06,DBLP:conf/icde/BohmOPY07,DBLP:journals/pvldb/TranMS20}, we choose parameter ranges that (i) cover multiple orders of magnitude, (ii) are normalized across datasets for comparability, and (iii) stress scalability under different workloads. Using percentages prevents parameter choices from being tied to a particular dataset scale and enables consistent comparison across datasets with different cardinalities. The default $(n,s,k)$ configuration represents a representative mid-range streaming workload with substantial window overlap and a moderate number of requested disjoint intervals.

\noindent
\textbf{Performance Metrics.}
Following prior evaluations~\cite{DBLP:journals/tkde/ZhuWYZW17,DBLP:conf/edbt/YangSRW11,DBLP:conf/sigmod/MouratidisBP06,DBLP:conf/cikm/ShastriDRW11,DBLP:journals/pacmmod/ZhuWYZ23,DBLP:journals/pvldb/TranFS16,DBLP:journals/pvldb/YoonLL19,DBLP:journals/tkde/TaoP06,DBLP:conf/icde/BohmOPY07,DBLP:journals/pvldb/TranMS20}, we use five metrics: $1)$ \emph{Overall running time}: total query and update time. $2)$ \emph{Data throughput}: updates processed per second. $3)$ \emph{Candidate set size}: average number of candidate maximal intervals. $4)$ \emph{Memory usage}: peak memory across algorithms and datasets. $5)$ \emph{Impact of number of partitions}: average running time per window under different partition counts.

\noindent
\textbf{Competitors.} 
We compare our partition-based algorithm with RT~\cite{DBLP:conf/ismb/RuzzoT99}, Tournament~\cite{DBLP:journals/ijfcs/BaeT07}, and their incremental streaming variants, Streaming RT and Streaming Tournament; implementation details appear in Appendixes~A and~B of our technical report~\cite{long-version}. All algorithms are implemented in C++ and evaluated on a Windows $11$ PC with an Intel Core i$5$-$13500$H CPU and $32$GB memory.

\subsection{Algorithm Comparisons}
\noindent
\textbf{Running time comparison}.
We assess the running time of all algorithms under varying data scales, ranging from $20\%$ to $100\%$ of each dataset. For example, using 20\% of \textsc{Stock} means that 20\% of the original records are selected to form a new dataset containing approximately $9$M $\times 0.2=1.8$M records.  
The running time required to process \query\ while simultaneously updating partitions for the entire dataset is reported in Figure~\ref{fig:overall}.
The results demonstrate that the partition-based approach consistently outperforms all competing algorithms across all datasets, achieving substantial performance gains. For example, on the \textsc{Stock} dataset, it achieves speedups of $8\times$ and $10\times$ over Streaming RT and Streaming Tournament, respectively. Additionally, it outperforms the original RT and Tournament by more than an order of magnitude. Similar advantages are observed across the remaining datasets, where the partition-based algorithm exhibits consistent and often significant superiority. 

\begin{figure}[t]
    \centering
            \includegraphics[width=\linewidth]{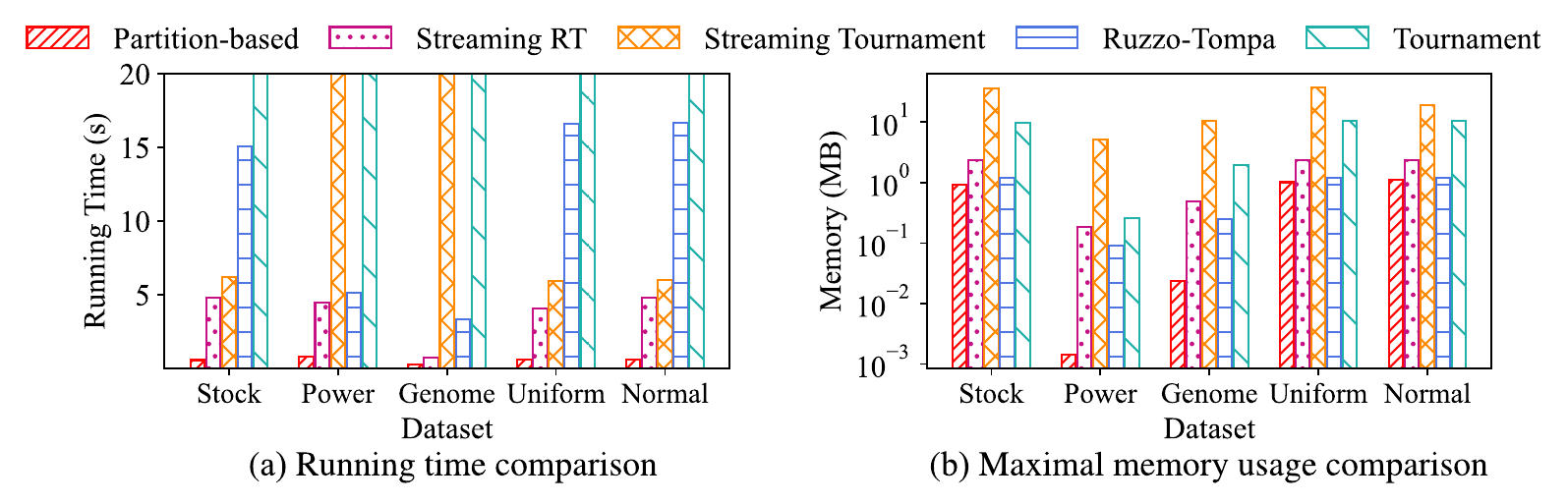}
            \caption{Performance comparison for datasets of 1M records.}
    \label{fig:1M_Time-memory}
\end{figure}

\begin{figure*}[t]
    \centering
            \includegraphics[width=\linewidth]{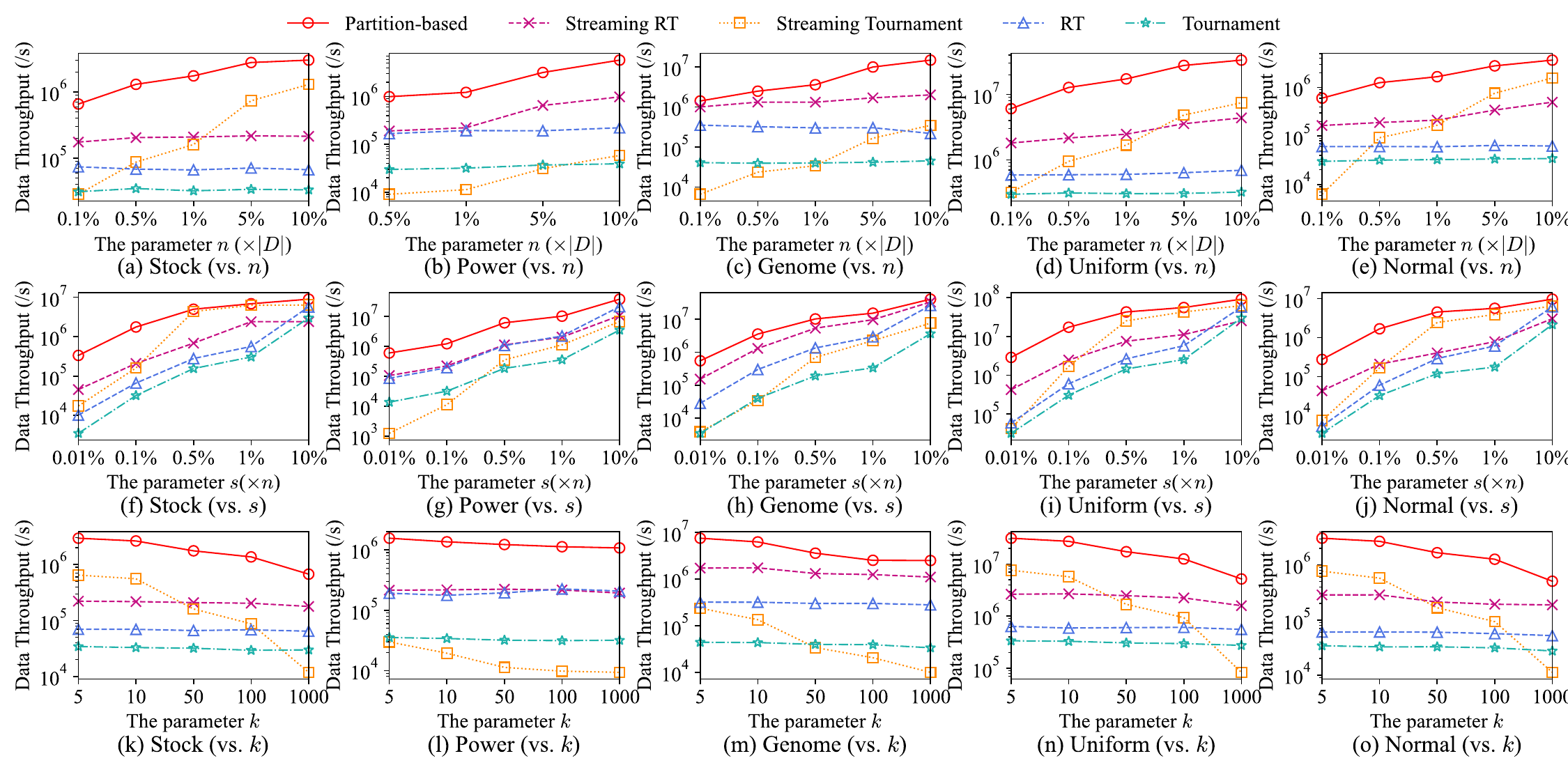} 
            \caption{Data throughput comparison under different datasets.}
    \label{fig:parameter}
\end{figure*}

The superior performance of our method can be primarily attributed to the proposed partitioning strategy. Compared with competitors, our method has three advantages. First, unlike Streaming RT, it identifies partitions likely to contain query results and confines maximal interval identification and comparison to them, substantially reducing the solution space. Second, unlike Streaming Tournament, our method avoids index structures and their restructuring overhead. Third, its partition-based incremental maintenance accurately identifies affected ranges and maximizes reuse of previously computed maximal intervals. In contrast, although both Streaming RT and Streaming Tournament support incremental maintenance, Streaming RT incurs a larger affected range due to the absence of partitions, while Streaming Tournament frequently requires expensive structural adjustments. Furthermore, without streaming incremental maintenance, original RT and Tournament must re-identify all maximal intervals in each window, yielding lower efficiency.

We also observe that the performance gap varies considerably across datasets. 
For example, on the \textsc{Stock} dataset, the partition-based approach outperforms Streaming Tournament by approximately an order of magnitude on average, whereas its advantage increases to nearly two orders of magnitude on the \textsc{Power} and \textsc{Genome} datasets. 

To further examine the sensitivity of the algorithms to data distribution, we analyze the running time required by each algorithm to process one million objects across multiple datasets. The results are presented in Figure~\ref{fig:1M_Time-memory}(a).
We observe that, for each algorithm, the time needed to process one million updates varies across datasets. This variation arises because maximal intervals, the targets of interest, are inherently sensitive to the underlying data distribution. As subsequences of signed values, both their length and their counts are influenced by distributional characteristics. 

However, it is evident that our approach exhibits the least sensitivity to data distribution. This robustness stems from the partitioning strategy: although the length and number of maximal intervals may fluctuate significantly across datasets, our method processes only a small subset of the data, specifically up to $k$ partitions most likely to contain query results. This effectively mitigates dataset-specific variations and reduces the impact of distributional differences. Consequently, we conclude that our approach demonstrates notably greater stability compared to all competing algorithms.

\noindent
\textbf{Data throughput comparison}. 
Figure~\ref{fig:parameter} presents data throughput of all algorithms under varying parameter settings across different datasets, showing that our approach consistently outperforms all competitors by a significant margin. As shown in Figures~\ref{fig:parameter}(a)-(e), as $n$ increases, the throughput of both our approach and Streaming Tournament improves, while the throughput of the other algorithms remains relatively stable. Note, due to the small size of the Power dataset, its $n$ value starts at 0.5\% rather than 0.1\%. A larger $n$ indicates that the query window contains more objects and thus more maximal intervals. However, our approach continues to process only the partitions likely to contain query results, while Streaming Tournament directly searches for the top-$k$ results based on the index. Consequently, the number of updates processed per unit time increases for both methods.

Conversely, Figures~\ref{fig:parameter}(k)-(o) show that as $k$ increases, the data throughput of the partition-based algorithm experiences a slight decrease, while that of Streaming Tournament decreases sharply. Meanwhile, the throughput of other algorithms remains relatively stable. This behavior arises because the partition-based algorithm must process more partitions as $k$ grows, whereas Streaming Tournament performs more queries. Both RT and Streaming RT retain all maximal intervals, and Tournament incurs substantial overhead due to frequent index restructuring. As a result, their performance is less sensitive to $k$. Importantly, even when $k$ is large, the partition-based algorithm still maintains a clear performance advantage.

Figures~\ref{fig:parameter}(f)-(j) illustrate that all algorithms exhibit increased throughput as $s$ becomes larger. When $s$ is small, there is considerable overlap between windows, causing many objects to be recomputed repeatedly. As $s$ increases, the overhead from redundant computations is reduced, leading to an overall improvement.

\noindent
\textbf{Candidate set size}. 
Table~\ref{tab:candid-set} reports the average number of maximal intervals maintained in the candidate set (i.e., $\mathcal{CM}$ in our algorithm). The two Tournament-Tree-based-approaches do not rely on a candidate set and are therefore excluded from this statistic. Both Streaming RT and RT maintain all maximal intervals, so their values in the table reflect the average number of maximal intervals per window. 
In contrast, our approach significantly reduces the number of intervals in the candidate set. This improvement arises because we incrementally maintain maximal intervals only within the partitions likely to contain query results. The \textsc{Power} dataset is the only exception: our method reduces candidate set only from 17 to 16, because this dataset inherently produces very few maximal intervals, leaving limited room for enhancement.

\renewcommand{\arraystretch}{0.7}
\begin{table}[t]
\footnotesize
  \caption{The average size of the candidate set (UNIT:$count$) 
  }
  \setlength{\tabcolsep}{2mm}
  \label{tab:candid-set}
\begin{tabular}{c c c c c c}
    \toprule
     Solutions      & \textsc{Stock}    & \textsc{Power}    &  \textsc{Genome} &  \textsc{Uniform} & \textsc{Normal}\\
    \midrule
    Partition-based & 930               & 16                &   194         &  1270     & 1252\\
    
    Streaming RT    & 1927              & 17                &   207         &  2961     & 3032\\

    RT              & 1927              & 17                &   207         &  2961     & 3032 \\                  \bottomrule
\end{tabular}
\end{table}

Table~\ref{tab:candid} presents the average per-window statistics for candidate count and usage.  Count denotes the number of maximal intervals identified, while usage counts how many times these intervals are used, including comparisons for top-$k$ selection and candidate-set maintenance. The results show that our approach achieves substantial improvements in both metrics.
This advantage stems from our partitioning strategy: during maximal interval identification, top-$k$ comparison, and incremental maintenance, the strategy avoids exhaustively processing all maximal intervals and instead focuses only on partitions likely to contain query results.

\noindent
\textbf{Memory usage comparison.}
Figure~\ref{fig:1M_Time-memory}(b) illustrates the maximum memory usage of all the algorithms during their execution across various datasets. Among these, the partition-based approach consumes the least memory, as the partitioning strategy minimizes the number of partitions processed. Both Streaming RT and RT require comparing all maximal intervals to select the top-$k$, necessitating the use of more memory to identify and maintain all maximal intervals. While Streaming Tournament and Tournament do not maintain a candidate set, they still incur higher memory costs compared to the other algorithms due to the construction of the Tournament tree. Furthermore, the memory usage of the partition-based approach is notably lower than that of the other algorithms, particularly on the \textsc{Power} dataset. A detailed examination of this dataset reveals that most maximal intervals exist in the form of the maximal interval with the largest sum within its respective partition. These intervals can be directly obtained during partition construction, eliminating the need to invoke the maximal interval identification algorithm (Algorithm~\ref{algo:non-maximum}), which results in a substantial reduction in memory usage.

\renewcommand{\arraystretch}{0.7}
\begin{table}[t]
\footnotesize
  \caption{Average count/usage statistics of candidates per window (UNIT:$count$ for count, and $times$ for usage)}
  \setlength{\tabcolsep}{3mm}
  \label{tab:candid}
  \begin{tabular}{lccccc}
    \toprule
     Dataset&Partition-based & Streaming RT  & RT  \\
    \midrule
    {\textsc{Stock}}  & 49.29 /  370.82        & 1244.15 / 1927.74    & 39401.89 / 1927.74   \\

   {\textsc{Power}} & 0.51 / 0.97           & 92.15 /17.13       & 431.62 / 17.13             \\
    {\textsc{Genome}} & 1.67 / 6.30           & 6.90 / 207.73         & 724.41 / 207.73       \\
    {\textsc{Uniform}} & 50.74 / 569.54           & 1228.59 / 2961.59       & 49977.34 / 2961.59      \\
    {\textsc{Normal}}  & 51.83 / 620.57           & 1137.16 / 3032.76        & 49972.99 / 3032.76        \\                              
  \bottomrule
\end{tabular}
\end{table}

\begin{table}[t]
\centering
\footnotesize
\caption{Average running time per sliding window under different numbers of partitions on \textsc{Stock} dataset.}
\label{tab:ave-run-time-per-window}
\setlength{\tabcolsep}{3.5pt}
\begin{tabular}{l c c c c c c}
\toprule
Groups by $|\mathcal{P}|$
& $1\mbox{-}100$
& $101\mbox{-}200$
& $201\mbox{-}300$
& $301\mbox{-}400$
& $401\mbox{-}500$
& $501\mbox{-}600$ \\
Groups' Frequency
& $42.95\%$
& $30.74\%$
& $17.92\%$
& $6.41\%$
& $1.68\%$
& $0.30\%$ \\
Time$/|W|$ ($\times 10^{-5}$s)
& 5.45
& 4.83
& 4.34
& 3.89
& 3.35
& 2.91 \\
\midrule
Baselines
& \multicolumn{2}{c}{Streaming RT}
& \multicolumn{2}{c}{Streaming TT}
& \multicolumn{1}{c}{RT}
& \multicolumn{1}{c}{TT} \\
Time$/|W|$ ($\times 10^{-5}$s)
& \multicolumn{2}{c}{43.56}
& \multicolumn{2}{c}{56.64}
& \multicolumn{1}{c}{137.36}
& \multicolumn{1}{c}{283.38} \\
\bottomrule
\end{tabular}
\end{table}

\noindent

\textbf{Impact of number of partitions.} 
Our partitioning scheme is data-driven: for a fixed window, the cutting points and thus the number of partitions $|\mathcal{P}|$ are uniquely determined by the distribution of signed values within the window. Consequently, $|\mathcal{P}|$ is not controlled by any algorithmic parameter, but varies naturally with the data.
To evaluate how the partition count affects performance in practice, we conduct a window stratification study under the default parameters. Specifically, we process the entire stream and record for every sliding window, then group windows into ranges according to partition counts $|\mathcal{P}|$ (e.g., 1-100, 101-200, etc).
For each group, we report both its frequency and the average running time per window. Table~\ref{tab:ave-run-time-per-window} reports the results on the \textsc{Stock} stream, together with the average running time per window required by other baselines. 
Note, Streaming TT and TT represent Streaming Tournament Tree and Tournament Tree, respectively. Please refer to our technical report~\cite{long-version} for results on other datasets.

The results reveal a clear monotonic trend: windows with larger $|\mathcal{P}|$ have lower average running time. This behavior reflects the key advantage of the partition-based framework. A larger number of partitions provides finer-grained localization and more effective pruning, allowing the top-$k$ search to focus on fewer relevant regions and thereby reducing overall computation. As a result, the reduction in search space outweighs
the relatively small overhead of maintaining additional partitions. Across all partition
ranges, our method consistently outperforms these baselines by a substantial margin. Importantly, even for windows with small $|\mathcal{P}|$, the running time remains stable and highly competitive, indicating that our method does not degrade under coarse partitioning.

\section{Related Work} \label{sec:related work}

To the best of our knowledge, no existing algorithm efficiently addresses the \query\ problem over streaming data. Prior studies mainly address maximal interval discovery in static sequences or consider different interval optimization objectives.

\noindent
\textbf{Maximal interval discovery in static sequences.}
Finding maximal sum intervals is a classic problem derived from the maximum subarray problem~\cite{online:Maximums_subarray_problem}. Existing solutions mainly fall into two categories. Scan-based approaches, such as the Ruzzo-Tompa (RT) algorithm~\cite{DBLP:conf/ismb/RuzzoT99}, identify all maximal intervals in linear time. However, RT always enumerates the full set of maximal intervals, even when only the top-$k$ results are needed, and the intervals are not ranked during enumeration, requiring an additional selection step. Index-based approaches, such as the Tournament-Tree method~\cite{DBLP:journals/ijfcs/BaeT07}, build a tree-structured index that enables direct retrieval of the highest-ranked interval. After each interval is retrieved, its elements are marked as holes to prevent reuse in subsequent results. While effective for static data, these approaches assume a fixed input sequence and are not designed for continuous updates.

\noindent

\noindent
\textbf{Streaming Interval Extraction.} 
Several studies investigate extracting interval-shaped patterns from data streams under different formulations. For example, \cite{DBLP:journals/ijar/MantenoglouAP23} proposes an online framework (oPIEC) for detecting probabilistic maximal intervals in noisy streams.
This approach operates on probability and locates the threshold-based intervals that may overlap. 
In contrast, \query\ operates on signed streams and returns globally ranked disjoint maximal intervals, making the problem fundamentally different.

\noindent
\textbf{Other Variants}. Some work~\cite{DBLP:journals/tcbb/Csuros04,DBLP:conf/cocoon/BengtssonC06,DBLP:conf/cpm/GawrychowskiN15} also studies selecting $k$ disjoint subsequences that maximize the total sum. 
This objective differs fundamentally from \query. For instance, given the sequence $(10, -9, 10$, $-11, 1)$ and $k=2$, these variants would select intervals $I[1,1]$ and $I[3,3]$ (total sum $= 20$), whereas \query\ selects $MaxI[1,3]$ and $MaxI[5,5]$).

\noindent

\textbf{Limitations of existing approaches.}
Although static KMaxI algorithms could in principle be adapted to streaming settings through incremental maintenance, such adaptations remain costly. 
Scan-based methods must maintain and repeatedly compare all maximal intervals within the sliding window to derive the top-$k$ results, leading to large candidate sets and high update overhead. 
For index-based methods require continuous structural updates to maintain correctness under window slides, including node insertions, deletions, and restoration of previously excluded elements. These operations introduce significant maintenance costs in high-velocity streams. 
In contrast, our method avoids both full enumeration and global indexing.
By exploiting the partition containment guarantee, we prune partitions that cannot contribute to the top-$k$ results and restrict interval identification and maintenance to a small number of promising partitions, enabling efficient and robust incremental processing over data streams.

\section{Conclusion}
This paper introduces an innovative partitioning strategy designed to support \query\ over streaming data. The proposed strategy demonstrates significant efficiency in candidate identification, candidate set maintenance, and query result retrieval. Its adaptability enables effective support for \query\ using a compact candidate set. Our comprehensive experiments across diverse datasets validate the superior performance and relative stability of the proposed algorithms. In the near future, we would like to study parallel algorithms for \query\ over streaming data, leveraging the inherent property that partitions can be processed independently.

\begin{acks}
 The work was partially supported by the National Key Research and Development Program of China (No. 2024YFF0617702); the National Natural Science Foundation of China (Nos. U22A2025, 62232007, and U23A20309); the 111 Project (No. B16009); the Ant Group Research Program (No. 2025021900003).
\end{acks}


\bibliographystyle{ACM-Reference-Format}
\bibliography{reference}


\begin{thebibliography}{38}


\ifx \showCODEN    \undefined \def \showCODEN     #1{\unskip}     \fi
\ifx \showDOI      \undefined \def \showDOI       #1{#1}\fi
\ifx \showISBNx    \undefined \def \showISBNx     #1{\unskip}     \fi
\ifx \showISBNxiii \undefined \def \showISBNxiii  #1{\unskip}     \fi
\ifx \showISSN     \undefined \def \showISSN      #1{\unskip}     \fi
\ifx \showLCCN     \undefined \def \showLCCN      #1{\unskip}     \fi
\ifx \shownote     \undefined \def \shownote      #1{#1}          \fi
\ifx \showarticletitle \undefined \def \showarticletitle #1{#1}   \fi
\ifx \showURL      \undefined \def \showURL       {\relax}        \fi
\providecommand\bibfield[2]{#2}
\providecommand\bibinfo[2]{#2}
\providecommand\natexlab[1]{#1}
\providecommand\showeprint[2][]{arXiv:#2}

\bibitem[\protect\citeauthoryear{??}{onl}{2025a}]%
        {online:Alpha_Vantage}
 \bibinfo{year}{2025}\natexlab{a}.
\newblock \bibinfo{booktitle}{\emph{Alpha Vantage}}.
\newblock
\urldef\tempurl%
\url{https://www.alphavantage.co/}
\showURL{%
\tempurl}
\newblock
\shownote{Accessed: October 15, 2025.}


\bibitem[\protect\citeauthoryear{??}{onl}{2025b}]%
        {online:Electrical}
 \bibinfo{year}{2025}\natexlab{b}.
\newblock \bibinfo{booktitle}{\emph{Electrical Grid Half Hourly (UK)}}.
\newblock
\urldef\tempurl%
\url{https://www.kaggle.com/datasets/thedevastator/gb-electrical-grid-half-hourly-data-2008-present}
\showURL{%
\tempurl}
\newblock
\shownote{Accessed: October 15, 2025.}


\bibitem[\protect\citeauthoryear{??}{onl}{2025c}]%
        {online:genome}
 \bibinfo{year}{2025}\natexlab{c}.
\newblock \bibinfo{booktitle}{\emph{Haemophilus influenzae whole genome}}.
\newblock
\urldef\tempurl%
\url{https://www.ncbi.nlm.nih.gov/Traces/wgs/JAMLFG01}
\showURL{%
\tempurl}
\newblock
\shownote{Accessed: October 15, 2025.}


\bibitem[\protect\citeauthoryear{??}{onl}{2025d}]%
        {online:Maximums_subarray_problem}
 \bibinfo{year}{2025}\natexlab{d}.
\newblock \bibinfo{booktitle}{\emph{Maximum subarray problem}}.
\newblock
\urldef\tempurl%
\url{https://en.wikipedia.org/wiki/Maximum_subarray_problem}
\showURL{%
\tempurl}
\newblock
\shownote{Accessed: October 6, 2025.}


\bibitem[\protect\citeauthoryear{Ajtai, Jayram, Kumar, and Sivakumar}{Ajtai et~al\mbox{.}}{2002}]%
        {DBLP:conf/stoc/AjtaiJKS02}
\bibfield{author}{\bibinfo{person}{Mikl{\'{o}}s Ajtai}, \bibinfo{person}{T.~S. Jayram}, \bibinfo{person}{Ravi Kumar}, {and} \bibinfo{person}{D. Sivakumar}.} \bibinfo{year}{2002}\natexlab{}.
\newblock \showarticletitle{Approximate counting of inversions in a data stream}. In \bibinfo{booktitle}{\emph{Proceedings on 34th Annual {ACM} Symposium on Theory of Computing, May 19-21, 2002, Montr{\'{e}}al, Qu{\'{e}}bec, Canada}}, \bibfield{editor}{\bibinfo{person}{John~H. Reif}} (Ed.). \bibinfo{publisher}{{ACM}}, \bibinfo{pages}{370--379}.
\newblock
\urldef\tempurl%
\url{https://doi.org/10.1145/509907.509964}
\showDOI{\tempurl}


\bibitem[\protect\citeauthoryear{Bae and Takaoka}{Bae and Takaoka}{2007}]%
        {DBLP:journals/ijfcs/BaeT07}
\bibfield{author}{\bibinfo{person}{Sung~Eun Bae} {and} \bibinfo{person}{Tadao Takaoka}.} \bibinfo{year}{2007}\natexlab{}.
\newblock \showarticletitle{Algorithms for k-Disjoint Maximum Subarrays}.
\newblock \bibinfo{journal}{\emph{Int. J. Found. Comput. Sci.}} \bibinfo{volume}{18}, \bibinfo{number}{2} (\bibinfo{year}{2007}), \bibinfo{pages}{319--339}.
\newblock
\urldef\tempurl%
\url{https://doi.org/10.1142/S012905410700470X}
\showDOI{\tempurl}


\bibitem[\protect\citeauthoryear{Bengtsson and Chen}{Bengtsson and Chen}{2006}]%
        {DBLP:conf/cocoon/BengtssonC06}
\bibfield{author}{\bibinfo{person}{Fredrik Bengtsson} {and} \bibinfo{person}{Jingsen Chen}.} \bibinfo{year}{2006}\natexlab{}.
\newblock \showarticletitle{Computing Maximum-Scoring Segments in Almost Linear Time}. In \bibinfo{booktitle}{\emph{Computing and Combinatorics, 12th Annual International Conference, {COCOON} 2006, Taipei, Taiwan, August 15-18, 2006, Proceedings}} \emph{(\bibinfo{series}{Lecture Notes in Computer Science})}, \bibfield{editor}{\bibinfo{person}{Danny~Z. Chen} {and} \bibinfo{person}{D.~T. Lee}} (Eds.), Vol.~\bibinfo{volume}{4112}. \bibinfo{publisher}{Springer}, \bibinfo{pages}{255--264}.
\newblock
\urldef\tempurl%
\url{https://doi.org/10.1007/11809678\_28}
\showDOI{\tempurl}


\bibitem[\protect\citeauthoryear{B{\"{o}}hm, Ooi, Plant, and Yan}{B{\"{o}}hm et~al\mbox{.}}{2007}]%
        {DBLP:conf/icde/BohmOPY07}
\bibfield{author}{\bibinfo{person}{Christian B{\"{o}}hm}, \bibinfo{person}{Beng~Chin Ooi}, \bibinfo{person}{Claudia Plant}, {and} \bibinfo{person}{Ying Yan}.} \bibinfo{year}{2007}\natexlab{}.
\newblock \showarticletitle{Efficiently Processing Continuous k-NN Queries on Data Streams}. In \bibinfo{booktitle}{\emph{Proceedings of the 23rd International Conference on Data Engineering, {ICDE} 2007, The Marmara Hotel, Istanbul, Turkey, April 15-20, 2007}}, \bibfield{editor}{\bibinfo{person}{Rada Chirkova}, \bibinfo{person}{Asuman Dogac}, \bibinfo{person}{M.~Tamer {\"{O}}zsu}, {and} \bibinfo{person}{Timos~K. Sellis}} (Eds.). \bibinfo{publisher}{{IEEE} Computer Society}, \bibinfo{pages}{156--165}.
\newblock
\urldef\tempurl%
\url{https://doi.org/10.1109/ICDE.2007.367861}
\showDOI{\tempurl}


\bibitem[\protect\citeauthoryear{Chen, Lv, Yu, and Liu}{Chen et~al\mbox{.}}{2017}]%
        {DBLP:journals/dase/ChenLYL17}
\bibfield{author}{\bibinfo{person}{Ben Chen}, \bibinfo{person}{Zhijin Lv}, \bibinfo{person}{Xiaohui Yu}, {and} \bibinfo{person}{Yang Liu}.} \bibinfo{year}{2017}\natexlab{}.
\newblock \showarticletitle{Sliding Window Top-K Monitoring over Distributed Data Streams}.
\newblock \bibinfo{journal}{\emph{Data Sci. Eng.}} \bibinfo{volume}{2}, \bibinfo{number}{4} (\bibinfo{year}{2017}), \bibinfo{pages}{289--300}.
\newblock
\urldef\tempurl%
\url{https://doi.org/10.1007/S41019-017-0053-1}
\showDOI{\tempurl}


\bibitem[\protect\citeauthoryear{Cs{\"{u}}r{\"{o}}s}{Cs{\"{u}}r{\"{o}}s}{2004}]%
        {DBLP:journals/tcbb/Csuros04}
\bibfield{author}{\bibinfo{person}{Mikl{\'{o}}s Cs{\"{u}}r{\"{o}}s}.} \bibinfo{year}{2004}\natexlab{}.
\newblock \showarticletitle{Maximum-Scoring Segment Sets}.
\newblock \bibinfo{journal}{\emph{{IEEE} {ACM} Trans. Comput. Biol. Bioinform.}} \bibinfo{volume}{1}, \bibinfo{number}{4} (\bibinfo{year}{2004}), \bibinfo{pages}{139--150}.
\newblock
\urldef\tempurl%
\url{https://doi.org/10.1109/TCBB.2004.43}
\showDOI{\tempurl}


\bibitem[\protect\citeauthoryear{Escobar, Betancur, and Isaac}{Escobar et~al\mbox{.}}{2024}]%
        {Escobar2024}
\bibfield{author}{\bibinfo{person}{Eros~D. Escobar}, \bibinfo{person}{Daniel Betancur}, {and} \bibinfo{person}{Idi~A. Isaac}.} \bibinfo{year}{2024}\natexlab{}.
\newblock \showarticletitle{Optimal Power and Battery Storage Dispatch Architecture for Microgrids: Implementation in a Campus Microgrid}.
\newblock \bibinfo{journal}{\emph{Smart Grids and Sustainable Energy}} \bibinfo{volume}{9}, \bibinfo{number}{2} (\bibinfo{year}{2024}), \bibinfo{pages}{27}.
\newblock
\showISSN{2731-8087}
\urldef\tempurl%
\url{https://doi.org/10.1007/s40866-024-00210-8}
\showDOI{\tempurl}


\bibitem[\protect\citeauthoryear{Gawrychowski and Nicholson}{Gawrychowski and Nicholson}{2015}]%
        {DBLP:conf/cpm/GawrychowskiN15}
\bibfield{author}{\bibinfo{person}{Pawel Gawrychowski} {and} \bibinfo{person}{Patrick~K. Nicholson}.} \bibinfo{year}{2015}\natexlab{}.
\newblock \showarticletitle{Encodings of Range Maximum-Sum Segment Queries and Applications}. In \bibinfo{booktitle}{\emph{Combinatorial Pattern Matching - 26th Annual Symposium, {CPM} 2015, Ischia Island, Italy, June 29 - July 1, 2015, Proceedings}} \emph{(\bibinfo{series}{Lecture Notes in Computer Science})}, \bibfield{editor}{\bibinfo{person}{Ferdinando Cicalese}, \bibinfo{person}{Ely Porat}, {and} \bibinfo{person}{Ugo Vaccaro}} (Eds.), Vol.~\bibinfo{volume}{9133}. \bibinfo{publisher}{Springer}, \bibinfo{pages}{196--206}.
\newblock
\urldef\tempurl%
\url{https://doi.org/10.1007/978-3-319-19929-0\_17}
\showDOI{\tempurl}


\bibitem[\protect\citeauthoryear{Gedik, Wu, Yu, and Liu}{Gedik et~al\mbox{.}}{2007}]%
        {DBLP:conf/icde/GedikWYL07}
\bibfield{author}{\bibinfo{person}{Bugra Gedik}, \bibinfo{person}{Kun{-}Lung Wu}, \bibinfo{person}{Philip~S. Yu}, {and} \bibinfo{person}{Ling Liu}.} \bibinfo{year}{2007}\natexlab{}.
\newblock \showarticletitle{A Load Shedding Framework and Optimizations for M-way Windowed Stream Joins}. In \bibinfo{booktitle}{\emph{Proceedings of the 23rd International Conference on Data Engineering, {ICDE} 2007, The Marmara Hotel, Istanbul, Turkey, April 15-20, 2007}}, \bibfield{editor}{\bibinfo{person}{Rada Chirkova}, \bibinfo{person}{Asuman Dogac}, \bibinfo{person}{M.~Tamer {\"{O}}zsu}, {and} \bibinfo{person}{Timos~K. Sellis}} (Eds.). \bibinfo{publisher}{{IEEE} Computer Society}, \bibinfo{pages}{536--545}.
\newblock
\urldef\tempurl%
\url{https://doi.org/10.1109/ICDE.2007.367899}
\showDOI{\tempurl}


\bibitem[\protect\citeauthoryear{Gu, Yu, and Wang}{Gu et~al\mbox{.}}{2007}]%
        {DBLP:conf/icde/GuYW07}
\bibfield{author}{\bibinfo{person}{Xiaohui Gu}, \bibinfo{person}{Philip~S. Yu}, {and} \bibinfo{person}{Haixun Wang}.} \bibinfo{year}{2007}\natexlab{}.
\newblock \showarticletitle{Adaptive Load Diffusion for Multiway Windowed Stream Joins}. In \bibinfo{booktitle}{\emph{Proceedings of the 23rd International Conference on Data Engineering, {ICDE} 2007, The Marmara Hotel, Istanbul, Turkey, April 15-20, 2007}}, \bibfield{editor}{\bibinfo{person}{Rada Chirkova}, \bibinfo{person}{Asuman Dogac}, \bibinfo{person}{M.~Tamer {\"{O}}zsu}, {and} \bibinfo{person}{Timos~K. Sellis}} (Eds.). \bibinfo{publisher}{{IEEE} Computer Society}, \bibinfo{pages}{146--155}.
\newblock
\urldef\tempurl%
\url{https://doi.org/10.1109/ICDE.2007.367860}
\showDOI{\tempurl}


\bibitem[\protect\citeauthoryear{Haghani, Michel, and Aberer}{Haghani et~al\mbox{.}}{2009}]%
        {DBLP:conf/cikm/HaghaniMA09}
\bibfield{author}{\bibinfo{person}{Parisa Haghani}, \bibinfo{person}{Sebastian Michel}, {and} \bibinfo{person}{Karl Aberer}.} \bibinfo{year}{2009}\natexlab{}.
\newblock \showarticletitle{Evaluating top-k queries over incomplete data streams}. In \bibinfo{booktitle}{\emph{Proceedings of the 18th {ACM} Conference on Information and Knowledge Management, {CIKM} 2009, Hong Kong, China, November 2-6, 2009}}, \bibfield{editor}{\bibinfo{person}{David~Wai{-}Lok Cheung}, \bibinfo{person}{Il{-}Yeol Song}, \bibinfo{person}{Wesley~W. Chu}, \bibinfo{person}{Xiaohua Hu}, {and} \bibinfo{person}{Jimmy Lin}} (Eds.). \bibinfo{publisher}{{ACM}}, \bibinfo{pages}{877--886}.
\newblock
\urldef\tempurl%
\url{https://doi.org/10.1145/1645953.1646064}
\showDOI{\tempurl}


\bibitem[\protect\citeauthoryear{Jin, Yi, Chen, Yu, and Lin}{Jin et~al\mbox{.}}{2010}]%
        {DBLP:journals/vldb/JinYCYL10}
\bibfield{author}{\bibinfo{person}{Cheqing Jin}, \bibinfo{person}{Ke Yi}, \bibinfo{person}{Lei Chen}, \bibinfo{person}{Jeffrey~Xu Yu}, {and} \bibinfo{person}{Xuemin Lin}.} \bibinfo{year}{2010}\natexlab{}.
\newblock \showarticletitle{Sliding-window top-\emph{k} queries on uncertain streams}.
\newblock \bibinfo{journal}{\emph{{VLDB} J.}} \bibinfo{volume}{19}, \bibinfo{number}{3} (\bibinfo{year}{2010}), \bibinfo{pages}{411--435}.
\newblock
\urldef\tempurl%
\url{https://doi.org/10.1007/S00778-009-0171-0}
\showDOI{\tempurl}


\bibitem[\protect\citeauthoryear{Jin, Lai, Hoshyar, Innab, Shutaywi, Deebani, and Swathi}{Jin et~al\mbox{.}}{2025}]%
        {jin2025ideally}
\bibfield{author}{\bibinfo{person}{Yuan Jin}, \bibinfo{person}{Yunliang Lai}, \bibinfo{person}{Azadeh~Noori Hoshyar}, \bibinfo{person}{Nisreen Innab}, \bibinfo{person}{Meshal Shutaywi}, \bibinfo{person}{Wejdan Deebani}, {and} \bibinfo{person}{A Swathi}.} \bibinfo{year}{2025}\natexlab{}.
\newblock \showarticletitle{An ideally designed deep trust network model for heart disease prediction based on seagull optimization and Ruzzo Tompa algorithm}.
\newblock \bibinfo{journal}{\emph{Scientific Reports}} \bibinfo{volume}{15}, \bibinfo{number}{1} (\bibinfo{year}{2025}), \bibinfo{pages}{6035}.
\newblock


\bibitem[\protect\citeauthoryear{Karlin and Altschul}{Karlin and Altschul}{1993}]%
        {doi:10.1073/pnas.90.12.5873}
\bibfield{author}{\bibinfo{person}{S Karlin} {and} \bibinfo{person}{S~F Altschul}.} \bibinfo{year}{1993}\natexlab{}.
\newblock \showarticletitle{Applications and statistics for multiple high-scoring segments in molecular sequences.}
\newblock \bibinfo{journal}{\emph{Proceedings of the National Academy of Sciences}} \bibinfo{volume}{90}, \bibinfo{number}{12} (\bibinfo{year}{1993}), \bibinfo{pages}{5873--5877}.
\newblock
\urldef\tempurl%
\url{https://doi.org/10.1073/pnas.90.12.5873}
\showDOI{\tempurl}
\showeprint{https://www.pnas.org/doi/pdf/10.1073/pnas.90.12.5873}


\bibitem[\protect\citeauthoryear{Karlin and Brendel}{Karlin and Brendel}{1992}]%
        {doi:10.1126/science.1621093}
\bibfield{author}{\bibinfo{person}{Samuel Karlin} {and} \bibinfo{person}{Volker Brendel}.} \bibinfo{year}{1992}\natexlab{}.
\newblock \showarticletitle{Chance and Statistical Significance in Protein and DNA Sequence Analysis}.
\newblock \bibinfo{journal}{\emph{Science}} \bibinfo{volume}{257}, \bibinfo{number}{5066} (\bibinfo{year}{1992}), \bibinfo{pages}{39--49}.
\newblock
\urldef\tempurl%
\url{https://doi.org/10.1126/science.1621093}
\showDOI{\tempurl}
\showeprint{https://www.science.org/doi/pdf/10.1126/science.1621093}


\bibitem[\protect\citeauthoryear{Li, Gu, Zhou, Ma, and Yu}{Li et~al\mbox{.}}{2017}]%
        {DBLP:conf/edbt/Li0ZMY17}
\bibfield{author}{\bibinfo{person}{Tianyi Li}, \bibinfo{person}{Yu Gu}, \bibinfo{person}{Xiangmin Zhou}, \bibinfo{person}{Qian Ma}, {and} \bibinfo{person}{Ge Yu}.} \bibinfo{year}{2017}\natexlab{}.
\newblock \showarticletitle{An Effective and Efficient Truth Discovery Framework over Data Streams}. In \bibinfo{booktitle}{\emph{Proceedings of the 20th International Conference on Extending Database Technology, {EDBT} 2017, Venice, Italy, March 21-24, 2017}}, \bibfield{editor}{\bibinfo{person}{Volker Markl}, \bibinfo{person}{Salvatore Orlando}, \bibinfo{person}{Bernhard Mitschang}, \bibinfo{person}{Periklis Andritsos}, \bibinfo{person}{Kai{-}Uwe Sattler}, {and} \bibinfo{person}{Sebastian Bre{\ss}}} (Eds.). \bibinfo{publisher}{OpenProceedings.org}, \bibinfo{pages}{180--191}.
\newblock
\urldef\tempurl%
\url{https://doi.org/10.5441/002/EDBT.2017.17}
\showDOI{\tempurl}


\bibitem[\protect\citeauthoryear{Liang, Ren, Weerkamp, Meij, and de~Rijke}{Liang et~al\mbox{.}}{2014}]%
        {10.1145/2661829.2661905}
\bibfield{author}{\bibinfo{person}{Shangsong Liang}, \bibinfo{person}{Zhaochun Ren}, \bibinfo{person}{Wouter Weerkamp}, \bibinfo{person}{Edgar Meij}, {and} \bibinfo{person}{Maarten de Rijke}.} \bibinfo{year}{2014}\natexlab{}.
\newblock \showarticletitle{Time-Aware Rank Aggregation for Microblog Search}. In \bibinfo{booktitle}{\emph{Proceedings of the 23rd ACM International Conference on Conference on Information and Knowledge Management}} (Shanghai, China) \emph{(\bibinfo{series}{CIKM '14})}. \bibinfo{publisher}{Association for Computing Machinery}, \bibinfo{address}{New York, NY, USA}, \bibinfo{pages}{989–998}.
\newblock
\showISBNx{9781450325981}
\urldef\tempurl%
\url{https://doi.org/10.1145/2661829.2661905}
\showDOI{\tempurl}


\bibitem[\protect\citeauthoryear{Mantenoglou, Artikis, and Paliouras}{Mantenoglou et~al\mbox{.}}{2023}]%
        {DBLP:journals/ijar/MantenoglouAP23}
\bibfield{author}{\bibinfo{person}{Periklis Mantenoglou}, \bibinfo{person}{Alexander Artikis}, {and} \bibinfo{person}{Georgios Paliouras}.} \bibinfo{year}{2023}\natexlab{}.
\newblock \showarticletitle{Online event recognition over noisy data streams}.
\newblock \bibinfo{journal}{\emph{Int. J. Approx. Reason.}}  \bibinfo{volume}{161} (\bibinfo{year}{2023}), \bibinfo{pages}{108993}.
\newblock
\urldef\tempurl%
\url{https://doi.org/10.1016/J.IJAR.2023.108993}
\showDOI{\tempurl}


\bibitem[\protect\citeauthoryear{Mouratidis, Bakiras, and Papadias}{Mouratidis et~al\mbox{.}}{2006}]%
        {DBLP:conf/sigmod/MouratidisBP06}
\bibfield{author}{\bibinfo{person}{Kyriakos Mouratidis}, \bibinfo{person}{Spiridon Bakiras}, {and} \bibinfo{person}{Dimitris Papadias}.} \bibinfo{year}{2006}\natexlab{}.
\newblock \showarticletitle{Continuous monitoring of top-k queries over sliding windows}. In \bibinfo{booktitle}{\emph{Proceedings of the {ACM} {SIGMOD} International Conference on Management of Data, Chicago, Illinois, USA, June 27-29, 2006}}, \bibfield{editor}{\bibinfo{person}{Surajit Chaudhuri}, \bibinfo{person}{Vagelis Hristidis}, {and} \bibinfo{person}{Neoklis Polyzotis}} (Eds.). \bibinfo{publisher}{{ACM}}, \bibinfo{pages}{635--646}.
\newblock
\urldef\tempurl%
\url{https://doi.org/10.1145/1142473.1142544}
\showDOI{\tempurl}


\bibitem[\protect\citeauthoryear{Pasternack and Roth}{Pasternack and Roth}{2009}]%
        {10.1145/1526709.1526840}
\bibfield{author}{\bibinfo{person}{Jeff Pasternack} {and} \bibinfo{person}{Dan Roth}.} \bibinfo{year}{2009}\natexlab{}.
\newblock \showarticletitle{Extracting article text from the web with maximum subsequence segmentation}. In \bibinfo{booktitle}{\emph{Proceedings of the 18th International Conference on World Wide Web}} (Madrid, Spain) \emph{(\bibinfo{series}{WWW '09})}. \bibinfo{publisher}{Association for Computing Machinery}, \bibinfo{address}{New York, NY, USA}, \bibinfo{pages}{971–980}.
\newblock
\showISBNx{9781605584874}


\bibitem[\protect\citeauthoryear{Ramp{\'{e}}rez, Zahmatkesh, and Valle}{Ramp{\'{e}}rez et~al\mbox{.}}{2021}]%
        {DBLP:conf/bigdataconf/RamperezZV21}
\bibfield{author}{\bibinfo{person}{V{\'{\i}}ctor Ramp{\'{e}}rez}, \bibinfo{person}{Shima Zahmatkesh}, {and} \bibinfo{person}{Emanuele~Della Valle}.} \bibinfo{year}{2021}\natexlab{}.
\newblock \showarticletitle{Scaling the monitoring of approximate top-k queries in streaming windows}. In \bibinfo{booktitle}{\emph{2021 {IEEE} International Conference on Big Data (Big Data), Orlando, FL, USA, December 15-18, 2021}}, \bibfield{editor}{\bibinfo{person}{Yixin Chen}, \bibinfo{person}{Heiko Ludwig}, \bibinfo{person}{Yicheng Tu}, \bibinfo{person}{Usama~M. Fayyad}, \bibinfo{person}{Xingquan Zhu}, \bibinfo{person}{Xiaohua Hu}, \bibinfo{person}{Suren Byna}, \bibinfo{person}{Xiong Liu}, \bibinfo{person}{Jianping Zhang}, \bibinfo{person}{Shirui Pan}, \bibinfo{person}{Vagelis Papalexakis}, \bibinfo{person}{Jianwu Wang}, \bibinfo{person}{Alfredo Cuzzocrea}, {and} \bibinfo{person}{Carlos Ordonez}} (Eds.). \bibinfo{publisher}{{IEEE}}, \bibinfo{pages}{181--189}.
\newblock
\urldef\tempurl%
\url{https://doi.org/10.1109/BIGDATA52589.2021.9672041}
\showDOI{\tempurl}


\bibitem[\protect\citeauthoryear{Ruzzo and Tompa}{Ruzzo and Tompa}{1999}]%
        {DBLP:conf/ismb/RuzzoT99}
\bibfield{author}{\bibinfo{person}{Walter~L. Ruzzo} {and} \bibinfo{person}{Martin Tompa}.} \bibinfo{year}{1999}\natexlab{}.
\newblock \showarticletitle{A Linear Time Algorithm for Finding All Maximal Scoring Subsequences}. In \bibinfo{booktitle}{\emph{Proceedings of the Seventh International Conference on Intelligent Systems for Molecular Biology, August 6-10, 1999, Heidelberg, Germany}}, \bibfield{editor}{\bibinfo{person}{Thomas Lengauer}, \bibinfo{person}{Reinhard Schneider}, \bibinfo{person}{Peer Bork}, \bibinfo{person}{Douglas~L. Brutlag}, \bibinfo{person}{Janice~I. Glasgow}, \bibinfo{person}{Hans{-}Werner Mewes}, {and} \bibinfo{person}{Ralf Zimmer}} (Eds.). \bibinfo{publisher}{{AAAI}}, \bibinfo{pages}{234--241}.
\newblock
\urldef\tempurl%
\url{http://www.aaai.org/Library/ISMB/1999/ismb99-027.php}
\showURL{%
\tempurl}


\bibitem[\protect\citeauthoryear{Shastri, Yang, Rundensteiner, and Ward}{Shastri et~al\mbox{.}}{2011}]%
        {DBLP:conf/cikm/ShastriDRW11}
\bibfield{author}{\bibinfo{person}{Avani Shastri}, \bibinfo{person}{Di Yang}, \bibinfo{person}{Elke~A. Rundensteiner}, {and} \bibinfo{person}{Matthew~O. Ward}.} \bibinfo{year}{2011}\natexlab{}.
\newblock \showarticletitle{MTopS: scalable processing of continuous top-k multi-query workloads}. In \bibinfo{booktitle}{\emph{Proceedings of the 20th {ACM} Conference on Information and Knowledge Management, {CIKM} 2011, Glasgow, United Kingdom, October 24-28, 2011}}, \bibfield{editor}{\bibinfo{person}{Craig Macdonald}, \bibinfo{person}{Iadh Ounis}, {and} \bibinfo{person}{Ian Ruthven}} (Eds.). \bibinfo{publisher}{{ACM}}, \bibinfo{pages}{1107--1116}.
\newblock
\urldef\tempurl%
\url{https://doi.org/10.1145/2063576.2063736}
\showDOI{\tempurl}


\bibitem[\protect\citeauthoryear{Spouge, Mari{\~n}o-Ram{\'\i}rez, and Sheetlin}{Spouge et~al\mbox{.}}{2012}]%
        {spouge2012ruzzo}
\bibfield{author}{\bibinfo{person}{John~L Spouge}, \bibinfo{person}{Leonardo Mari{\~n}o-Ram{\'\i}rez}, {and} \bibinfo{person}{Sergey~L Sheetlin}.} \bibinfo{year}{2012}\natexlab{}.
\newblock \showarticletitle{The ruzzo-tompa algorithm can find the maximal paths in weighted, directed graphs on a one-dimensional lattice}. In \bibinfo{booktitle}{\emph{2012 IEEE 2nd International Conference on Computational Advances in Bio and medical Sciences (ICCABS)}}. IEEE, \bibinfo{pages}{1--6}.
\newblock


\bibitem[\protect\citeauthoryear{Subercaze, Gravier, Gillani, Kammoun, and Laforest}{Subercaze et~al\mbox{.}}{2017}]%
        {DBLP:journals/pvldb/SubercazeGGKL17}
\bibfield{author}{\bibinfo{person}{Julien Subercaze}, \bibinfo{person}{Christophe Gravier}, \bibinfo{person}{Syed Gillani}, \bibinfo{person}{Abderrahmen Kammoun}, {and} \bibinfo{person}{Fr{\'{e}}d{\'{e}}rique Laforest}.} \bibinfo{year}{2017}\natexlab{}.
\newblock \showarticletitle{Upsortable: Programming TopK Queries Over Data Streams}.
\newblock \bibinfo{journal}{\emph{Proc. {VLDB} Endow.}} \bibinfo{volume}{10}, \bibinfo{number}{12} (\bibinfo{year}{2017}), \bibinfo{pages}{1873--1876}.
\newblock
\urldef\tempurl%
\url{https://doi.org/10.14778/3137765.3137797}
\showDOI{\tempurl}


\bibitem[\protect\citeauthoryear{Tao and Papadias}{Tao and Papadias}{2006}]%
        {DBLP:journals/tkde/TaoP06}
\bibfield{author}{\bibinfo{person}{Yufei Tao} {and} \bibinfo{person}{Dimitris Papadias}.} \bibinfo{year}{2006}\natexlab{}.
\newblock \showarticletitle{Maintaining Sliding Window Skylines on Data Streams}.
\newblock \bibinfo{journal}{\emph{{IEEE} Trans. Knowl. Data Eng.}} \bibinfo{volume}{18}, \bibinfo{number}{2} (\bibinfo{year}{2006}), \bibinfo{pages}{377--391}.
\newblock
\urldef\tempurl%
\url{https://doi.org/10.1109/TKDE.2006.48}
\showDOI{\tempurl}


\bibitem[\protect\citeauthoryear{Tran, Fan, and Shahabi}{Tran et~al\mbox{.}}{2016}]%
        {DBLP:journals/pvldb/TranFS16}
\bibfield{author}{\bibinfo{person}{Luan~V. Tran}, \bibinfo{person}{Liyue Fan}, {and} \bibinfo{person}{Cyrus Shahabi}.} \bibinfo{year}{2016}\natexlab{}.
\newblock \showarticletitle{Distance-based Outlier Detection in Data Streams}.
\newblock \bibinfo{journal}{\emph{Proc. {VLDB} Endow.}} \bibinfo{volume}{9}, \bibinfo{number}{12} (\bibinfo{year}{2016}), \bibinfo{pages}{1089--1100}.
\newblock
\urldef\tempurl%
\url{https://doi.org/10.14778/2994509.2994526}
\showDOI{\tempurl}


\bibitem[\protect\citeauthoryear{Tran, Mun, and Shahabi}{Tran et~al\mbox{.}}{2020}]%
        {DBLP:journals/pvldb/TranMS20}
\bibfield{author}{\bibinfo{person}{Luan~V. Tran}, \bibinfo{person}{Minyoung Mun}, {and} \bibinfo{person}{Cyrus Shahabi}.} \bibinfo{year}{2020}\natexlab{}.
\newblock \showarticletitle{Real-Time Distance-Based Outlier Detection in Data Streams}.
\newblock \bibinfo{journal}{\emph{Proc. {VLDB} Endow.}} \bibinfo{volume}{14}, \bibinfo{number}{2} (\bibinfo{year}{2020}), \bibinfo{pages}{141--153}.
\newblock
\urldef\tempurl%
\url{https://doi.org/10.14778/3425879.3425885}
\showDOI{\tempurl}


\bibitem[\protect\citeauthoryear{Yang, Shastri, Rundensteiner, and Ward}{Yang et~al\mbox{.}}{2011}]%
        {DBLP:conf/edbt/YangSRW11}
\bibfield{author}{\bibinfo{person}{Di Yang}, \bibinfo{person}{Avani Shastri}, \bibinfo{person}{Elke~A. Rundensteiner}, {and} \bibinfo{person}{Matthew~O. Ward}.} \bibinfo{year}{2011}\natexlab{}.
\newblock \showarticletitle{An optimal strategy for monitoring top-k queries in streaming windows}. In \bibinfo{booktitle}{\emph{{EDBT} 2011, 14th International Conference on Extending Database Technology, Uppsala, Sweden, March 21-24, 2011, Proceedings}}, \bibfield{editor}{\bibinfo{person}{Anastasia Ailamaki}, \bibinfo{person}{Sihem Amer{-}Yahia}, \bibinfo{person}{Jignesh~M. Patel}, \bibinfo{person}{Tore Risch}, \bibinfo{person}{Pierre Senellart}, {and} \bibinfo{person}{Julia Stoyanovich}} (Eds.). \bibinfo{publisher}{{ACM}}, \bibinfo{pages}{57--68}.
\newblock
\urldef\tempurl%
\url{https://doi.org/10.1145/1951365.1951375}
\showDOI{\tempurl}


\bibitem[\protect\citeauthoryear{Yang, Xie, Zhang, and Chen}{Yang et~al\mbox{.}}{2025}]%
        {10999049}
\bibfield{author}{\bibinfo{person}{Meng Yang}, \bibinfo{person}{Rui Xie}, \bibinfo{person}{Yongjun Zhang}, {and} \bibinfo{person}{Yue Chen}.} \bibinfo{year}{2025}\natexlab{}.
\newblock \showarticletitle{Robust Microgrid Dispatch With Real-Time Energy Sharing and Endogenous Uncertainty}.
\newblock \bibinfo{journal}{\emph{IEEE Transactions on Smart Grid}} \bibinfo{volume}{16}, \bibinfo{number}{4} (\bibinfo{year}{2025}), \bibinfo{pages}{3085--3098}.
\newblock
\urldef\tempurl%
\url{https://doi.org/10.1109/TSG.2025.3569095}
\showDOI{\tempurl}


\bibitem[\protect\citeauthoryear{Yoon, Lee, and Lee}{Yoon et~al\mbox{.}}{2019}]%
        {DBLP:journals/pvldb/YoonLL19}
\bibfield{author}{\bibinfo{person}{Susik Yoon}, \bibinfo{person}{Jae{-}Gil Lee}, {and} \bibinfo{person}{Byung~Suk Lee}.} \bibinfo{year}{2019}\natexlab{}.
\newblock \showarticletitle{{NETS:} Extremely Fast Outlier Detection from a Data Stream via Set-Based Processing}.
\newblock \bibinfo{journal}{\emph{Proc. {VLDB} Endow.}} \bibinfo{volume}{12}, \bibinfo{number}{11} (\bibinfo{year}{2019}), \bibinfo{pages}{1303--1315}.
\newblock
\urldef\tempurl%
\url{https://doi.org/10.14778/3342263.3342269}
\showDOI{\tempurl}


\bibitem[\protect\citeauthoryear{Zhang, Yang, Zheng, Zhu, Li, and Wang}{Zhang et~al\mbox{.}}{2026}]%
        {long-version}
\bibfield{author}{\bibinfo{person}{Zhongshuai Zhang}, \bibinfo{person}{Xiaochun Yang}, \bibinfo{person}{Baihua Zheng}, \bibinfo{person}{Rui Zhu}, \bibinfo{person}{Haomin Li}, {and} \bibinfo{person}{Bin Wang}.} \bibinfo{year}{2026}\natexlab{}.
\newblock \bibinfo{title}{Continuous Query for Top-K Maximal Sum Intervals over Streaming Data (Long Version)}.
\newblock
\newblock
\urldef\tempurl%
\url{https://github.com/zhangzhongshuai/C-KMaxI/blob/main/Continuous_Query_for_Top-K_Maximal_Sum_Intervals_over_Streaming_Data_Long_Version.pdf}
\showURL{%
\tempurl}


\bibitem[\protect\citeauthoryear{Zhu, Wang, Yang, and Zheng}{Zhu et~al\mbox{.}}{2023}]%
        {DBLP:journals/pacmmod/ZhuWYZ23}
\bibfield{author}{\bibinfo{person}{Rui Zhu}, \bibinfo{person}{Bin Wang}, \bibinfo{person}{Xiaochun Yang}, {and} \bibinfo{person}{Baihua Zheng}.} \bibinfo{year}{2023}\natexlab{}.
\newblock \showarticletitle{Closest Pairs Search Over Data Stream}.
\newblock \bibinfo{journal}{\emph{Proc. {ACM} Manag. Data}} \bibinfo{volume}{1}, \bibinfo{number}{3} (\bibinfo{year}{2023}), \bibinfo{pages}{205:1--205:26}.
\newblock
\urldef\tempurl%
\url{https://doi.org/10.1145/3617326}
\showDOI{\tempurl}


\bibitem[\protect\citeauthoryear{Zhu, Wang, Yang, Zheng, and Wang}{Zhu et~al\mbox{.}}{2017}]%
        {DBLP:journals/tkde/ZhuWYZW17}
\bibfield{author}{\bibinfo{person}{Rui Zhu}, \bibinfo{person}{Bin Wang}, \bibinfo{person}{Xiaochun Yang}, \bibinfo{person}{Baihua Zheng}, {and} \bibinfo{person}{Guoren Wang}.} \bibinfo{year}{2017}\natexlab{}.
\newblock \showarticletitle{{SAP:} Improving Continuous Top-K Queries Over Streaming Data}.
\newblock \bibinfo{journal}{\emph{{IEEE} Trans. Knowl. Data Eng.}} \bibinfo{volume}{29}, \bibinfo{number}{6} (\bibinfo{year}{2017}), \bibinfo{pages}{1310--1328}.
\newblock
\urldef\tempurl%
\url{https://doi.org/10.1109/TKDE.2017.2662236}
\showDOI{\tempurl}


\end{thebibliography}

\end{document}